\documentclass{amsart}
\allowdisplaybreaks
\usepackage{slashed}
\usepackage{enumerate}
\usepackage{tikz}
\usetikzlibrary{arrows}
\usetikzlibrary{decorations.markings}
\usetikzlibrary{patterns}
\usepackage{mathrsfs}
\usepackage{amssymb}
\usepackage{upgreek}
\usepackage{hyperref}
\usepackage{comment}
\usepackage[margin=1.25in,dvips]{geometry}
\def\f12{\frac 1 2}

\makeindex

\newtheorem{definition}{Definition}[section]

\newtheorem{remark}{Remark}[section]
\newtheorem{lemma}{Lemma}[section]
\newtheorem{theorem}{Theorem}
\newtheorem{introtheorem}{Theorem}
\newtheorem{oldtheorem}{Theorem}[section]
\newtheorem*{theorem*}{Theorem}
\newtheorem*{corollary*}{Corollary}
\newtheorem{proposition}{Proposition}[section]

\newtheorem{corollary}{Corollary}[section]

\title[Time-Translation Invariance of Scattering and the Blue-Shift]{Time-Translation Invariance of Scattering Maps and Blue-Shift Instabilities on Kerr Black Hole Spacetimes}
\author{Mihalis Dafermos}
\address{\small University of Cambridge, Department of Pure Mathematics and Mathematical
Statistics, Wilberforce~Road,~Cambridge~CB3~0WA,~United~Kingdom\vskip.2pc}
\email{m.dafermos@dpmms.cam.ac.uk}
\author{Yakov Shlapentokh-Rothman}
\address{\small Princeton University, Department of Mathematics, Fine~Hall,~Washington~Road,~Princeton,~NJ~08544,~United~States\vskip.2pc }
\email{dafermos@math.princeton.edu}
\email{yshlapen@math.princeton.edu}
\date\today

\begin{document}
\maketitle

\begin{abstract}
In this paper, we provide an elementary, unified treatment of two distinct blue-shift
instabilities for the scalar wave equation on a fixed Kerr black hole background:
the celebrated blue-shift at the Cauchy horizon (familiar from the
strong cosmic censorship conjecture) and the time-reversed
red-shift at the event horizon  (relevant in classical scattering theory).

Our first theorem concerns the latter and constructs solutions to the wave equation on Kerr spacetimes such that the radiation field along the future event horizon vanishes and the radiation field along future null infinity decays at an arbitrarily fast polynomial rate, yet, the local energy of the solution is infinite near any point on the future event horizon.
Our second theorem constructs solutions to the wave equation on rotating Kerr spacetimes such that the radiation field along the past event horizon (extended into the black hole) vanishes and the radiation field along past null infinity decays at an arbitrarily fast polynomial rate, yet, the local energy of the solution is infinite near any point on the Cauchy horizon. 

The results make essential use
of the scattering theory developed in [M. Dafermos, I. Rodnianski and Y. Shlapentokh-Rothman \emph{A scattering theory for the wave equation on Kerr black hole exteriors}, preprint (2014) available at \url{http://arxiv.org/abs/1412.8379}] and 
exploit directly the time-translation invariance of the scattering map  and the non-triviality
of the transmission map.
\end{abstract}

\section{Introduction}

The blue-shift instability associated to the Kerr Cauchy horizon is one
of the most celebrated features of black hole interior geometries (see~\cite{penrose,wald}): 

Consider two freely falling observers,
$A$ and $B$, where $A$ enters the black hole, crossing the Cauchy horizon $\mathcal{CH}^+$ at finite proper
time, while $B$ remains forever outside, sending a pulse at constant frequency to $A$. 
The frequency of the pulses, as measured by $A$, become infinitely shifted to the blue as $A$ approaches $\mathcal{CH}^+$.

This situation is depicted below in the classic Penrose diagram representation.

\bigskip

\begin{center}
\begin{tikzpicture}

\fill[lightgray] (-1.25,1.25)--(1.25,-1.25) -- (.5,-2) --(-2,-2) to [out =90 , in = -120] (-1.25,1.25); 
\draw[dashed] (0,0) -- (1.25,-1.25) node[sloped,above,midway]{$\mathcal{I}^+$}; 
\draw[dashed] (1.25,-1.25) -- (.5,-2) node[sloped,below,midway]{$\mathcal{I}^-$}; 
\draw (-2,-2) -- (0,0) node[sloped,above,pos = .3]{$\mathcal{H}^+$}; 
\draw (-1.25,1.25) -- (0,0) node[sloped,above,pos = .6]{$\mathcal{CH}^+$}; 
\path [draw=black,fill=white] (0,0) circle (1/16); 
\path [draw=black,fill=white] (1.25,-1.25) circle (1/16); 
\draw [decoration={markings,mark=at position 1 with {\arrow[scale=1.25,]{>}}},
    postaction={decorate}] (0,-2) to[out = 75,in = -75] coordinate [pos = .1] (A0) coordinate [pos = .4] (A1) coordinate [pos = .6] (A2) coordinate [pos = .75] (A3) coordinate [pos = .85] (A4) coordinate [pos = .91] (A5) node[below,midway]{$\ \ \ B$} (.03,-.1);
\draw [thick,decoration={markings,mark=at position 1 with {\arrow[scale=1.25,]{>}}},
    postaction={decorate},dashed](A0) -- ++(135:1.8cm); 
\draw [thick,decoration={markings,mark=at position 1 with {\arrow[scale=1.25,]{>}}},
    postaction={decorate},dashed](A1) -- ++(135:1.91cm); 
\draw [thick,decoration={markings,mark=at position 1 with {\arrow[scale=1.25,]{>}}},
    postaction={decorate},dashed](A2) -- ++(135:1.89cm); 
\draw [thick,decoration={markings,mark=at position 1 with {\arrow[scale=1.25,]{>}}},
    postaction={decorate},dashed](A3) -- ++(135:1.8cm); 
\draw [thick,decoration={markings,mark=at position 1 with {\arrow[scale=1.25,]{>}}},
    postaction={decorate},dashed](A4) -- ++(135:1.7cm); 
    \draw [thick,decoration={markings,mark=at position 1 with {\arrow[scale=1.25,]{>}}},
    postaction={decorate},dashed](A5) -- ++(135:1.66cm); 
\draw [decoration={markings,mark=at position 1 with {\arrow[scale=1.25,]{>}}},
    postaction={decorate}] (-1,-2) to [out = 105, in =-105] node[pos = .7,below]{$A\ \ \ \ $}(-1,1.3) ;

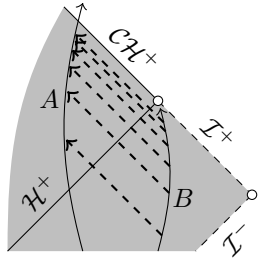
\node [align = flush center] at (0,-2.5) {Figure 1: The blue-shift effect at the Kerr Cauchy horizon};
\end{tikzpicture}
\end{center}

One can understand the shift in frequency using the time-translation invariance with respect to the ``stationary'' Killing field $T$. For this, note first that if $B$ is freely falling, then
the difference 
in proper time $s_B$ measured by observer $B$ is comparable to the difference in $t$ in that
$\frac{\partial t}{\partial s_B} \sim 1$. (In the limit where    $B$ is pushed to
past null infinity $\mathcal{I}^-$, this becomes an equality.)
On the other hand, the difference in proper time $s_A$ of $A$ satisfies 
$\frac{\partial t}{\partial s_A}\sim  \left(\overline{s}_A-s_A\right)^{-1}$, where $\overline{s}_A$ is the proper time when $A$ crosses $\mathcal{CH}^+$.  If $A$ and $B$ are chosen appropriately, then time-translation invariance will imply that the $t$-difference
between  two successive pulses along $A$ is always comparable to $1$.  One infers immediately an infinite blue-shift as $A$ crosses $\mathcal{CH}^+$.

It is precisely this instability which originally motivated Penrose~\cite{penrose} to conjecture his ``strong cosmic censorship'', according to which
Cauchy horizons are non-generic  in evolution for the Einstein equations.

The above blue-shift  can be contrasted with the equally celebrated red-shift effect on the event
horizon, which was discussed as early as 1939~\cite{OS}. 
Now it is observer $A$ who sends a signal at a constant frequency while crossing
the event horizon $\mathcal{H}^+$,
and this frequency is infinitely shifted to the red when received by $B$ (see also the textbook~\cite{wald}).

This situation is depicted below.
\begin{center}
\begin{tikzpicture}
\fill[lightgray] (0,0)--(1/2,-1/2)--(-1,-2) -- (-1.5,-1.5) -- (0,0); 
\draw[dashed] (0,0) -- (1/2,-1/2) node[sloped,above,midway]{$\mathcal{I}^+$}; 
\draw (-1.5,-1.5) -- (0,0) node[sloped,above,pos = .6]{$\mathcal{H}^+$}; 

\path [draw=black,fill=white] (0,0) circle (1/16); 
\draw [decoration={markings,mark=at position 1 with {\arrow[scale=1.25,]{>}}},
    postaction={decorate}] (.1,-.9) to[out = 80,in = -75]  node[below,midway]{$\ \ \ B$} (.03,-.07);

\draw [decoration={markings,mark=at position 1 with {\arrow[scale=1.25,]{>}}},
    postaction={decorate}] (-1,-2) to [out = 105, in =-100] coordinate [pos = .1] (A0) coordinate [pos = .4] (A1) coordinate [pos = .7] (A2) node[pos = .7,below]{$A\ \ \ \ $}(-1.1,-.9) ;
\draw [thick,decoration={markings,mark=at position 1 with {\arrow[scale=1.25,]{>}}},
    postaction={decorate},dashed](A0) -- ++(45:1.64cm); 
\draw [thick,decoration={markings,mark=at position 1 with {\arrow[scale=1.25,]{>}}},
    postaction={decorate},dashed](A1) -- ++(45:1.68cm); 
\draw [thick,decoration={markings,mark=at position 1 with {\arrow[scale=1.25,]{>}}},
    postaction={decorate},dashed](A2) -- ++(45:1.59cm); 

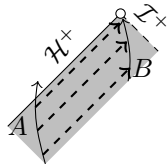
\node [align = flush center] at (-.5,-2.5) {Figure 2: The red-shift effect at the event horizon};
\end{tikzpicture}
\end{center}
Again, the effect can be understood using only the time-translation invariance,
in view again of the fact that $\frac{\partial t}{\partial s_A}\sim  \left(\mathring{s}_A-s_A\right)^{-1}$, where $\mathring{s}_A$ is now the proper time when $A$ crosses $\mathcal{H}^+$.

The red-shift is an important part of the conjectured stability of black holes in
evolution for the Einstein equations; see the discussion in~\cite{lectnotes}. On the other hand, it is clear that were one
to  reverse the time-orientation and consider backwards evolution, 
the above red-shift effect turns into  
a blue-shift, very analogous to the previous.
\begin{center}
\begin{tikzpicture}
\fill[lightgray] (0,0)--(1/2,-1/2)--(-1,-2) -- (-1.5,-1.5) -- (0,0); 
\draw[dashed] (0,0) -- (1/2,-1/2) node[sloped,above,midway]{$\mathcal{I}^+$}; 
\draw (-1.5,-1.5) -- (0,0) node[sloped,above,pos = .6]{$\mathcal{H}^+$}; 

\path [draw=black,fill=white] (0,0) circle (1/16); 
\draw [decoration={markings,mark=at position 1 with {\arrow[scale=1.25,]{>}}},
    postaction={decorate}] (.1,-.9) to[out = 80,in = -75]  coordinate [pos = .3] (A0) coordinate [pos = .6] (A1) coordinate [pos = .9] (A2) node[below,midway]{$\ \ \ B$} (.03,-.07);

\draw [decoration={markings,mark=at position 1 with {\arrow[scale=1.25,]{>}}},
    postaction={decorate}] (-1,-2) to [out = 105, in =-100]  node[pos = .7,below]{$A\ \ \ \ $}(-1.1,-.9) ;
\draw [thick,decoration={markings,mark=at position 1 with {\arrow[scale=1.25,]{>}}},
    postaction={decorate},dashed](A0) -- ++(-135:1.64cm); 
\draw [thick,decoration={markings,mark=at position 1 with {\arrow[scale=1.25,]{>}}},
    postaction={decorate},dashed](A1) -- ++(-135:1.68cm); 
\draw [thick,decoration={markings,mark=at position 1 with {\arrow[scale=1.25,]{>}}},
    postaction={decorate},dashed](A2) -- ++(-135:1.65cm); 

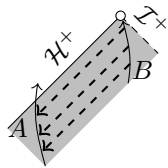
\node [align = flush center] at (-.5,-2.5) {Figure 3: The blue-shift effect in backwards evolution at the event horizon};
\end{tikzpicture}

\end{center}

This backwards blue-shift 
has been discussed in the context of classical scattering theory for the
wave equation on fixed black hole backgrounds~\cite{nicolas, scatter} as well as a scattering theory construction for
dynamic spacetimes satisfying  the non-linear vacuum Einstein 
equations~\cite{dynamic}. This phenomenon is also familiar from the well-known 
calculation of the Hawking effect and is related to the mechanism of 
particle creation. (See \cite{hawking, backgammon} for overview and discussion and~\cite{bachelot,
haefhawk} for recent mathematical results.)

The goal of the present paper is to  provide a unified treatment of the above two
blue-shift instabilities
in the context of the scattering theory for solutions to the wave equation
\begin{equation}\label{wave}
\Box_g\psi = 0.
\end{equation} 
Let us note that the Gaussian beam approximation~\cite{sbiergauss} already gives a sequence
of finite energy
solutions  for which the local energy near the event or Cauchy horizon, respectively, becomes arbitrarily
large. Here,
we would like to produce solutions parametrised  by smooth ``scattering data'' on future and past null infinity $\mathcal{I}^{\pm}$ respectively (with
polynomially decaying tails)  for which this local energy
is in fact \emph{infinite}.

Before stating our first result capturing the blue-shift instability
at the event horizon $\mathcal{H}^+$, we note that
a consequence of our previous work on scattering~\cite{scatter} is that there do indeed exist 
unique solutions $\psi$ to the wave equation~\eqref{wave} with prescribed radiation fields
on $\mathcal{H}^+\cup \mathcal{I}^+$. 
These solutions by construction
lie in a space defined by the finiteness of a \emph{degenerate}  energy at $\mathcal{H}^+$ and are smooth
away from the horizon. Our  result
(see Section~\ref{blowupEHsec} for the precise statement) can be roughly stated as follows.

\begin{introtheorem}[rough statement]
\label{theorem1}
On any subextremal Kerr spacetime,
there exists an axisymmetric solution $\psi$ to the wave equation such that the radiation field of $\psi$ along the future event horizon $\mathcal{H}^+$ vanishes and the radiation field of $\psi$ along future null infinity $\mathcal{I}^+$ decays at an arbitrarily fast polynomial rate, yet the local energy of $\psi$ is infinite in a neighbourhood of 
any point on the future event horizon $\mathcal{H}^+$.
\end{introtheorem}

To state our second result capturing the blue-shift instability at the Cauchy horizon, 
we note first that applying~\cite{scatter} and the well-posedness
of a characteristic initial value problem, there exist unique solutions to the
wave equation~\eqref{wave} with prescribed radiation fields on $\mathcal{H}^-\cup \mathcal{I}^-$,
where $\mathcal{H}^-$ is now extended to the black hole region.
Our second result (see Section~\ref{blowupCHsec} for the precise statement) 
concerns the blue-shift instability associated with the Cauchy horizon and
can be roughly stated as follows:

\begin{introtheorem}[rough statement]
\label{theorem2}
On any subextremal, rotating Kerr spacetime there exists an axisymmetric solution $\psi$ to the wave equation such that the radiation field of $\psi$ along the past event horizon $\mathcal{H}^-$ (extended into the black hole) vanishes and the radiation field of $\psi$ along past null infinity $\mathcal{I}^-$ decays at an arbitrarily fast polynomial rate, yet the local energy of $\psi$ is infinite in a neighbourhood of any point on the ingoing part of the Cauchy horizon $\mathcal{CH}^+$.
\end{introtheorem}

The statements of the two theorems are illustrated below:

\begin{center}
\begin{tikzpicture}
\fill[lightgray] (0,0) -- (1.5,-1.5) -- (0,-3) -- (-1.5,-1.5) -- (0,0);
\draw (0,0) -- node[sloped,above, scale = .75]{vanishing data} node[pos = .4,sloped,below,scale = .75]{$\mathcal{H}^+$} (-1.5,-1.5) --  (0,-3);
\draw[dashed] (0,0) -- node[sloped,above, scale = .75]{polynomial tail}  node[sloped,below,scale = .75]{$\mathcal{I}^+$}(1.5,-1.5) -- (0,-3);
\draw [fill = gray,domain=-135:45] plot ({-1+.2*cos(\x)}, {-1+.2*sin(\x)});
\draw [->] (-.7,-1.4) -- (-.95,-1.05) node[pos = .05,right,scale =.6]{$\psi$ has infinite energy};
\path [draw=black,fill=white] (0,0) circle (1/16);
\path [draw=black,fill=white] (0,-3) circle (1/16);
\path [draw=black,fill=white] (1.5,-1.5) circle (1/16);
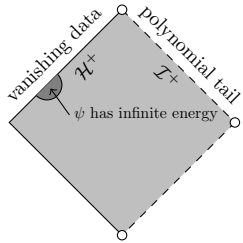
\node at (0,-3.5) {Figure 4: Illustration of Theorem 1};

\fill[lightgray] (8.75,0) -- (10.25,-1.5) -- (8.75,-3) -- (7.25,-1.5) -- (8.75,0);
\draw (8.75,0) -- (7.25,-1.5) -- node[sloped,above,scale = .75]{$\mathcal{H}^-$} (8.75,-3);
\draw [fill = lightgray] (7.25,-1.5) -- (5.75,0) -- (7.25,1.5) -- (8.75,0);
\draw (7.25,1.5) -- node[sloped,above,scale = .75]{$\mathcal{CH^+}$}(8.75,0);
\draw [fill = gray,domain=135:-45] plot ({8.25+.2*cos(-\x-90)}, {.5+.2*sin(-\x-90)});
\draw (5.75,0) -- node[sloped,above,scale = .75]{$\mathcal{CH^+}$}(7.25,1.5);
\draw[dashed] (8.75,0) -- (10.25,-1.5) -- node[sloped,below, scale = .75]{polynomial tail}  node[sloped,above,scale = .75]{$\mathcal{I}^-$}(8.75,-3);
\draw (5.75,0) -- node[sloped,below,scale = .75]{vanishing data}  (8.75,-3); 
\path [draw=black,fill=white] (8.75,0) circle (1/16);
\path [draw=black,fill=white] (8.75,-3) circle (1/16);
\path [draw=black,fill=white] (10.25,-1.5) circle (1/16);
\path [draw=black,fill=white] (5.75,0) circle (1/16);
\draw [->] (7.95,.1) -- (8.2,.45) node[pos = .05,left,scale =.6]{$\psi$ has infinite energy};
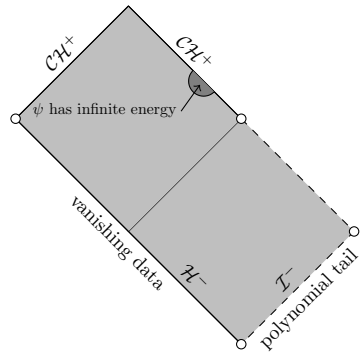
\node at (8,-3.5) {Figure 5: Illustration of Theorem 2};
\end{tikzpicture}
\end{center}

As motivation for considering the blow-up of local energy of $\psi$, recall that one consequence of  Christodoulou's formulation of the strong cosmic censorship conjecture~\cite{christ} is that generic perturbations of the Kerr spacetime should produce metrics which are \emph{not} extendible in $H^1_{\rm loc}$ through the Cauchy horizons. (In view of the work~\cite{Luk-D}, however, the Cauchy horizons will persist as null hypersurfaces through which the metric is $C^0$ extendible.)
We note that due to the linearity of the wave equation, it is an immediate consequence of both the above theorems that the singular behaviour 
along $\mathcal{H}^+$ or $\mathcal{CH}^+$ in fact holds for \emph{generic} solutions to the wave equation whose radiation fields along $\mathcal{I}^+$ or $\mathcal{I}^-$ respectively decay at arbitrarily fast polynomially rates.
We shall not try to formalise this further here, however.

\subsection{Previous and related work}
Theorem~\ref{theorem1}
above concerning the event horizon $\mathcal{H}^+$ is related to a similar result proven in~\cite{scatter} for the Schwarzschild $a=0$
case using   certain monotonicity properties of the wave equation in spherical symmetry
(Theorem 11.1~of~\cite{scatter}). 
In particular, the above generalisation shows that 
Theorem  2 of~\cite{scatter} holds for all $|a|<M$.

Our Theorem~\ref{theorem1} above concerning the Cauchy horizon $\mathcal{CH}^+$
 can be thought to complete a paper of McNamara~\cite{mcnamara},
where a conditional proof of a slightly weaker statement was given, showing
that $\psi$ failed to be $C^1$ at the horizon, subject however to verifying
certain statements concerning ``non-zero transmission'' to the Cauchy horizon (these needed statements in fact follow from our Theorem~\ref{tothecauchyandbeyond} and could be used to complete his proof). Our proof (see 
Section~\ref{sketchsec} below) will however be different from that of
McNamara.

Another approach to capturing the blue-shift instability at the Cauchy horizon 
$\mathcal{CH}^+$ is to identify
a condition on the solution \emph{along the event horizon $\mathcal{H}^+$} which
ensures blow up at $\mathcal{CH}^+$. 
Such a condition was given by Luk--Oh~\cite{lukoh} in the Reissner--Nordstr\"{o}m case, who moreover
showed that their condition indeed holds for \emph{generic} solutions arising
from compactly supported data posed on an asymptotically flat,
\emph{spacelike} hypersurface. 
(For some partial results concerning self-gravitating spherically symmetric scalar fields see~\cite{dafermos1,dafermos2}.) In addition, the work~\cite{lukoh} gives an explicit characterization
of the genericity assumption in terms of the asymptotics along future
null infinity $\mathcal{I}^+$.
In parallel with the present paper, 
Luk--Sbierski~\cite{luksbier} have obtained a Kerr analogue of the
result of~\cite{lukoh} relating 
a polynomial lower bound along $\mathcal{H}^+$ to infinite local energy
at $\mathcal{CH}^+$. Obtaining a characterization of spacelike initial data 
for which this lower bound holds remains an open problem. In broad terms, one expects that the strategy of~\cite{luksbier} will be applicable in the study
of the instability properties of the Cauchy horizon in the full non-linear theory
governed by the Einstein vacuum equations (cf.~\cite{Luk-D}).

  Let us also note several other classical attempts in the physics literature to understand the blue-shift instability at the Cauchy horizon~\cite{chanhart,gsns,mzs}. 
  
In the case of \emph{extremal} Reissner--Nordstr\"{o}m or \emph{extremal} Kerr, the local red-shift along $\mathcal{H}^+$ and the local blue-shift along $\mathcal{CH}^+$ both vanish. This leads to fundamentally different expectations for the qualitative behavior of waves in the black hole exterior and interior, see~\cite{aretakis1,aretakis2,aretakis3,aretakis4,dejan,dejan2} for the current state of the art. We note, however, that even the question of boundedness for general solutions to~\eqref{wave} on extremal Kerr {exteriors} remains an open problem (see the discussion in~\cite{partIII}). 

\subsection{Time-translation invariance of scattering theory and the approach of this paper}
\label{sketchsec}
The proofs we shall offer in this paper do not overtly rely on the
geometric optics approximation.
Rather, the idea is to exploit the relation between time-translation invariance and the blue-shift
directly  at the level of the wave equation~\eqref{wave}.
The argument is quite
soft and will not require precise 
decay estimates of the form used in~\cite{dafermos2, lukoh}
or in Theorem~11.1~of~\cite{scatter}. 
The fundamental three ingredients of scattering theory for our argument are
\begin{enumerate}
	\item[(a)] the existence of ``transmission maps'' taking initial data on $\mathcal{I}^-$ to a well defined radiation field along $\mathcal{CH}^+$ and initial data on $\mathcal{I}^+$ to a well defined radiation field on $\mathcal{H}^-$,
	\item[(b)] the commutation properties of the transmission maps with translation by the stationary Killing vector field $T$,
	\item[(c)] the non-vanishing of the transmission map.
\end{enumerate}

For Theorem~\ref{theorem1}, 
the necessary scattering theory results for (a), (b), (c)  follow \emph{a fortiori} from our recent works~\cite{partIII,scatter} which, among other things, establish a complete scattering theory for general solutions to the wave equation on Kerr exterior spacetimes in the full sub-extremal range $|a| < M$. See~\cite{bw,dk1,dk2,nicolas} for earlier results about scattering on Schwarzschild and \cite{fhm} for a textbook introduction to black hole scattering. 
For Theorem~\ref{theorem2}, 
however, the analogue of the scattering theory of~\cite{scatter} is not available (cf.~however the pioneering work~\cite{mcnamara2} and the more recent~\cite{anne, anne2} for the statements of uniform boundedness and continuous extendibility to the Cauchy horizons  for sufficiently regular solutions to the wave equation on both the Reissner--Nordstr\"{o}m and Kerr spacetimes). We shall prove by hand what we need (see Theorem~\ref{tothecauchyandbeyond}), but we hope that this paper further 
motivates a full treatment of interior black hole scattering, which would be of 
significant  independent interest.

The relevance of the commutation properties of the transmission map
for the blue-shift instability stems from the relation 
\begin{equation}
\label{relationhere}
\nabla_K K=\kappa K,
\end{equation}
satisfied by the Killing generator $K$ of the event or Cauchy horizon, which 
is a linear combination of the 
stationary vector field $T$  referred to above
and the generator $Z$ of axisymmetry.
Here $\kappa$ is the surface gravity.
In the case of Theorem~\ref{theorem1}, we will use the above 
properties of our scattering theory to construct solutions $\psi$
whose induced radiation fields on $\mathcal{H}^-$ are non-trivial and satisfy a scaling property
with respect to $T$, and thus, by axisymmetry of $\psi$, also with respect to $K$.
In view of the relation $(\ref{relationhere})$, this scaling property
necessitates that the local energy of $\psi$ is infinite at the sphere of bifurcation
where $\mathcal{H}^+$ and $\mathcal{H}^-$ meet and
where the vector field $K$ vanishes. The result then follows by
a form of propagation of singularities along $\mathcal{H}^+$. A similar argument
is used for Theorem~\ref{theorem2}.

\subsection{Outline}
We close  with an outline of the rest of this paper. 
In Section~\ref{reviewKerr} below we review the Kerr spacetime, defining the regions
of interest and setting our notation. Then, in Section~\ref{basic} we review some basic facts about the wave equation. We shall recall the pertinent scattering theory results
of~\cite{scatter} for the exterior
region in Section~\ref{reviewScat}. The precise formulations of our main two results are given
in Section~\ref{secResults}.  The proof of Theorem~\ref{theorem1}
will then be given
in Section~\ref{secPfThm1}. Finally, in Section~\ref{secPfThm2}, we
develop the scattering theory we need for the interior, using this to prove 
Theorem~\ref{theorem2}.

\subsection{Acknowledgements} We thank Igor Rodnianski for useful conversations and comments on the manuscript, and Jonathan Luk and Jan Sbierski for sharing details of their work~\cite{luksbier} with us. The second author also thanks Andr\'{e} Lisibach for helpful comments about the presentation. MD acknowledges support through NSF grant DMS-1405291 and EPSRC grant EP/K00865X/1. YS acknowledges support through an NSF Postdoctoral Research Fellowship under award no.\ 1502569.

\section{Review of the Kerr spacetime}\label{reviewKerr}
In order to fix our notation we will review the differentiable structure and metric of the Kerr spacetime as well as various coordinate systems. For a true introduction to the Kerr spacetime, we recommend~\cite{oneil} and~\cite{wald}.

The sub-extremal Kerr family has two free parameters $(a,M)$ which are required to satisfy $|a| < M$. Theorem 1 concerns the entire range $|a| < M$ while Theorem 2 requires that the rotation parameter $a \neq 0$. For convenience, however, throughout this section we will assume $0 < |a| < M$. In Remark~\ref{schw} we will note the part of the construction relevant for the Schwarzschild case ($a = 0$).

We begin by defining some useful constants $r_{\pm}$, $\kappa_{\pm}$ and functions $r^* : \mathbb{R}\setminus\{r_+,r_-\} \to \mathbb{R}$, $A : \mathbb{R}\setminus\{r_+,r_-\} \to \mathbb{R}$, $\Delta : \mathbb{R} \to \mathbb{R}$, $\rho^2 : \mathbb{R} \times (0,\pi) \to \mathbb{R}$, and $\Pi : \mathbb{R} \times (0,\pi) \to \mathbb{R}$:
\[r_{\pm} \doteq M \pm \sqrt{M^2-a^2},\qquad \kappa_{\pm} \doteq \frac{r_{\pm} - r_{\mp}}{2\left(r_{\pm}^2 + a^2\right)},\]
\begin{equation}\label{formrstar}
r^*\left(r\right) \doteq r + \frac{1}{2\kappa_+}\log\left(|r-r_+|\right) + \frac{1}{2\kappa_-}\log\left(\left|r-r_-\right|\right),
\end{equation}
\begin{equation*}
A\left(r\right) \doteq \frac{a}{r_+-r_-}\log\left(\left|\frac{r-r_+}{r-r_-}\right|\right),
\end{equation*}
\[\Delta\left(r\right) \doteq r^2-2Mr+a^2,\quad \rho^2\left(r,\theta\right) \doteq r^2+a^2\cos^2\theta,\quad \Pi\left(r,\theta\right) \doteq (r^2+a^2)^2 - a^2\sin^2\theta\Delta.\]
A calculation yields
\[\frac{dr^*}{dr} = \frac{r^2+a^2}{\Delta},\qquad \frac{dA}{dr} = \frac{a}{\Delta}.\]

Note that the two roots of $\Delta$ are given by $r_{\pm}$ and that our assumption that $|a| \in (0,M)$ guarantees that $r_{\pm}$ are distinct and hence that $\kappa_{\pm}$ are non-vanishing.
\subsection{Boyer--Lindquist coordinates in the exterior and interior}\label{blcoord}
We start by defining the \emph{exterior} region

\[\mathcal{D}_{\rm ext} \doteq \{(t,r,\theta,\varphi) \in \mathbb{R} \times [r_+,\infty) \times \mathbb{S}^2\}.\]

On $\mathcal{D}_{\rm ext}$ we define the Kerr metric by
\begin{equation}\label{metric}
g_{a,M} \doteq -\left(1-\frac{2Mr}{\rho^2}\right)dt^2 - \frac{4Mar\sin^2\theta}{\rho^2}dtd\varphi + \frac{\rho^2}{\Delta}dr^2 + \rho^2d\theta^2 + \sin^2\theta\frac{\Pi}{\rho^2}d\varphi^2,
\end{equation}

We next define the black hole \emph{interior} region
\[\mathcal{D}_{\rm int} \doteq \{(t,r,\theta,\varphi) \in \mathbb{R} \times (r_-,r_+) \times \mathbb{S}^2\}.\]
The metric $g_{a,M}$ on $\mathcal{D}_{\rm int}$ is defined by the same expression~\eqref{metric}.

Lastly, it will be useful to introduce the notations $T \doteq \partial_t$ and $Z \doteq \partial_{\varphi}$ for the stationary and axisymmetric Killing vector fields.

\subsection{Outgoing and ingoing Eddington--Finkelstein-like coordinates}
While the metric~\eqref{metric} takes a simple form in Boyer--Lindquist coordinates, these coordinates do not cover all portions of interest in the maximally extended Kerr spacetime; in particular, they do not cover the event horizons or Cauchy horizons.

We introduce a new coordinate system, ``outgoing Eddington--Finkelstein coordinates,'' by covering the region $\mathcal{D}_{\rm ext}$ with $(v,r,\theta,\varphi^*) \in \mathbb{R} \times (r_+,\infty) \times \mathbb{S}^2$, where the new $v$ and $\varphi^*$-coordinates are defined as follows
\begin{equation}\label{newvnewphi}
v\left(t,r\right) \doteq t + r^*(r),\qquad \varphi^*\left(\varphi,r\right) \doteq \varphi + A(r).
\end{equation}
The map $(t,r,\theta,\varphi)\mapsto (v,r,\theta,\varphi^*)$ is easily seen to be a diffeomorphism on $\mathcal{D}_{\rm ext}$, and under this change of coordinates, the metric~\eqref{metric} takes the form
\begin{align}\label{metric2}
g_{a,M} \doteq &-\left(1-\frac{2Mr}{\rho^2}\right)dv^2 +2dvdr- \frac{4Mar\sin^2\theta}{\rho^2}dvd\varphi^* + \rho^2d\theta^2
\\ \nonumber&\qquad + \sin^2\theta\frac{\Pi}{\rho^2}(d\varphi^*)^2-2a\sin^2\theta d\varphi^*dr.
\end{align}
The point of introducing these coordinates is that it is manifestly clear that the formula~\eqref{metric2} extends to define a Lorentzian manifold with boundary on $(\mathcal{D},g_{a,M})$, where
\begin{equation}\label{mathcalD}
\mathcal{D} \doteq \{(v,r,\theta,\varphi^*) \in \mathbb{R}\times [r_-,\infty)\times \mathbb{S}^2\}.
\end{equation}
Furthermore, inversion of the formulas~\eqref{newvnewphi} immediately yields an isometry between the Lorentzian manifold $(\mathcal{D} \cap \{r \in (r_-,r_+)\},g_{a,M})$ and $(\mathcal{D}_{\rm int},g_{a,M})$. Using this isometry, we now identify these two regions.

Along the hypersurface $\{r = r_+\}$ one may easily check that 
\begin{equation}\label{hawkvect}
K \doteq \partial_v + \frac{a}{2Mr_+}\partial_{\varphi^*}
\end{equation}
is null, future pointing,\footnote{We time-orient $\mathcal{D}_{\rm ext}$ with the timelike vector field $\nabla t$ and we time-orient $\mathcal{D}_{\rm in}$ with the timelike Boyer--Lindquist vector field $-\partial_r$. One may easily check these orientations agree along $\{r=r_+\}$.} and orthogonal to $\partial_{\theta}$ and $\partial_{\varphi^*}$; in particular, $\{r=r_+\}$ is a null hypersurface with future pointing null generator $\partial_v+\frac{a}{2Mr_+}\partial_{\varphi^*}$. We call $K$ the ``Hawking vector field''. We call $\{r = r_+\}$ the ``outgoing future event horizon'' and denote it by $\mathcal{H}^+_{\rm out}$. When there is no risk of confusion we may drop the ``outgoing'' and simply write ``future event horizon'' and $\mathcal{H}^+$. The event horizon forms the boundary between the exterior and interior regions. Similarly $\partial_v + \frac{a}{2Mr_-}\partial_{\varphi^*}$ is a future pointing null generator of $\{r=r_-\}$; we call $\{r = r_-\}$ the ``outgoing Cauchy horizon'' and denote it by $\mathcal{CH}_{\rm out}^+$.

Here is the Penrose diagram\footnote{The Penrose diagram can be formally defined as a depiction of the domain of yet another coordinate system $\left(U,V,\theta^A,\theta^B\right)$ associated to a ``double-null foliation''. See~\cite{pi} and~\cite{dejan3} for a construction of these coordinates. We only use these diagrams here as suggestive representations.}\label{doublenull} of the region $\mathcal{D}$:
\begin{center}
\begin{tikzpicture}
\fill[lightgray] (-2,2)--(2,-2)--(0,-4) -- (-4,0); 
\draw[dashed] (-2,2) -- (2,-2) --  (0,-4) -- (-4,0); 
\draw (-2,-2) -- (0,0) node[sloped,above,midway]{$\mathcal{H}^+_{\rm out}$}; 
\draw (-4,0) -- (-2,2) node[sloped,above,midway]{$\mathcal{CH}^+_{\rm out}$}; 
\path [draw=black,fill=white] (0,0) circle (1/16); 
\path [draw=black,fill=white] (2,-2) circle (1/16); 
\path [draw=black,fill=white] (0,-4) circle (1/16); 
\path [draw=black,fill=white] (-2,-2) circle (1/16); 
\path [draw=black,fill=white] (-4,0) circle (1/16); 
\path [draw=black,fill=white] (-2,2) circle (1/16); 
\draw (-2,0) node{$\mathcal{D}_{\rm int}$}; 
\draw (0,-2) node{$\mathcal{D}_{\rm ext}$}; 

\node [align = flush center] at (-1,-4.5) {Figure 6: The region $\mathcal{D}$};
\end{tikzpicture}
\end{center}

Note that the vector fields $T$ and $Z$ both smoothly extend to $\mathcal{D}$ and satisfy $T = \partial_v$ and $Z = \partial_{\varphi^*}$.

In order to define the ``past event horizon'' we need to introduce ``ingoing Eddington--Finkelstein coordinates''. These are $(u,r,\theta,\varphi_*) \in \mathbb{R} \times (r_+,\infty) \times \mathbb{S}^2$ coordinates defined on $\mathcal{D}_{{\rm ext}}$ in terms of Boyer--Lindquist coordinates by the formulas
\begin{equation}\label{newvnewphi2}
u\left(t,r\right) \doteq t - r^*(r), \qquad \varphi_*\left(\varphi,r\right) \doteq \varphi - A(r).
\end{equation}
The map $(t,r,\theta,\varphi) \mapsto (u,r,\theta,\varphi_*)$ is easily seen to be a diffeomorphism, and the metric~\eqref{metric} becomes
\begin{align}\label{metric3}
g_{a,M} \doteq &-\left(1-\frac{2Mr}{\rho^2}\right)du^2 -2dudr- \frac{4Mar\sin^2\theta}{\rho^2}dud\varphi^* + \rho^2d\theta^2
\\ \nonumber&\qquad + \sin^2\theta\frac{\Pi}{\rho^2}(d\varphi^*)^2+2a\sin^2\theta d\varphi^*dr.
\end{align}
Analogously to what we observed with the outgoing Eddington--Finkelstein coordinates, we now see that the metric expression~\eqref{metric3} defines a Lorentzian manifold with boundary on
\[\tilde{\mathcal{D}} \doteq \{(u,r,\theta,\varphi_*) \in \mathbb{R} \times [r_+,\infty)\times \mathbb{S}^2\}.\]
(We obviously could allow $r$ to take a larger set of values; these $r$-values however correspond to the interior of the \emph{white hole} region of the Kerr spacetime which will not concern us here.) We will call the hypersurface $\{r = r_+\}$ the ``past event horizon'' and denote it by $\mathcal{H}^-$. Similarly to $\mathcal{H}^+_{\rm out}$, $\mathcal{H}^-$ is a null hypersurface with future pointing null generator $\partial_u + \frac{a}{2Mr_+}\partial_{\varphi_*}$.

Note that as above the vector fields $T$ and $Z$ both smoothly extend to $\tilde{\mathcal{D}}$. In these coordinates they satisfy $T = \partial_u$ and $Z = \partial_{\varphi_*}$.

Finally, we can also introduce ingoing Eddington--Finkelstein coordinates in the region $\mathcal{D}_{\rm int}$ using the same formulas~\eqref{newvnewphi2}. As above, the corresponding metric expression~\eqref{metric3} defines a Lorentzian manifold with boundary on an extended manifold $\hat{\mathcal{D}} \doteq \{(u,r,\theta,\varphi_*) \in \mathbb{R} \times [r_-,r_+] \times \mathbb{S}^2\}$. We define the ``ingoing Cauchy horizon'' $\mathcal{CH}^+_{\rm in}$ to be the hypersurface $\{r = r_-\}$ and the ``ingoing future event horizon'' $\mathcal{H}^+_{\rm in}$ to be the hypersurface $\{r = r_+\}$. Again these are both null hypersurfaces, and one may take the future pointing null generators to be $-\partial_u - \frac{a}{2Mr_-}\partial_{\varphi_*}$ and $-\partial_u - \frac{a}{2Mr_+}\partial_{\varphi_*}$. Note the minus signs in the formulas for the null generators!

The vector fields $T$ and $Z$ smoothly extend to $\hat{\mathcal{D}}$ and satisfy $T = \partial_u$ and $Z = \partial_{\varphi_*}$.

\subsection{Kruskal type coordinates in the exterior and interior}
Lastly, following closely the presentation in~\cite{oneil}, we will discuss various Kruskal type coordinates. The importance of these coordinates is that they are valid near the bifurcation spheres of the event horizon and Cauchy horizon.

We start with the exterior. Define coordinates $\left(U_+,V_+,\theta,\varphi_+\right) \in (0,\infty)\times (0,\infty) \times \mathbb{S}^2$ on $\mathcal{D}_{\rm ext}$ in terms of Boyer--Lindquist coordinates by
\[
U_+\left(t,r\right) \doteq \exp\left(-\kappa_+\left(t-r^*(r)\right)\right),\qquad V_+\left(t,r\right) \doteq \exp\left(\kappa_+\left(t+r^*(r)\right)\right), \]
\[\varphi_+\left(\varphi,t\right) \doteq \varphi - \frac{a}{r_+^2+a^2}t.\]
The map $(t,r,\theta,\varphi) \mapsto (U_+,V_+,\theta,\varphi_+)$ is easily checked to be a diffeomorphism on $\mathcal{D}_{\rm ext}$. We will not give the explicit form of the metric $g_{a,M}$ in these coordinates (see~\cite{oneil} for the relevant formulas); instead we will just note that after doing so, it is manifestly clear that the Lorentzian manifold $(\mathcal{D}_{\rm ext},g_{a,M})$ extends a Lorentzian manifold with stratified boundary $(\overline{\mathcal{D}}_{\rm ext},g_{a,M})$ where
\[\overline{\mathcal{D}}_{\rm ext} \doteq \{\left(U_+,V_+,\theta,\varphi_+\right) \in [0,\infty) \times [0,\infty) \times \mathbb{S}^2\},\]
\[\partial\overline{\mathcal{D}}_{\rm ext} = \{V_+ = 0\} \cup \{U_+ = 0\}.\] 
Note that $\{V_+ = 0\}$ and $\{U_+ = 0\}$ intersect along a sphere.

Furthermore, the naturally defined diffeomorphism $(U_+,V_+,\theta,\varphi_+) \mapsto (v,r,\theta,\varphi^*)$ on $\mathcal{D}_{\rm ext}$ immediately extends to a diffeomorphism of $\{\left(U_+,V_+,\theta,\varphi_+\right) \in [0,\infty) \times (0,\infty) \times \mathbb{S}^2\}$ and $\mathcal{D} \cap [r_+,\infty)$ where, in particular, the future outgoing event horizon $\mathcal{H}^+_{\rm out}$ corresponds to the set $\{U_+ = 0\} \cap \{V_+ > 0\}$.  Similarly, the past event horizon $\mathcal{H}^-$ corresponds to $\{V_+ = 0\} \cap \{U_+ > 0\}$. The sphere $\{(U_+,V_+) = (0,0)\}$ will be denoted by $\mathcal{B}_+$ and is called the ``bifurcation sphere'' (of the event horizon). We will set
\begin{equation}\label{eventbif}
\overline{\mathcal{H}^+_{\rm out}} \doteq \{U_+ = 0\} \cap \{V_+ \geq 0\},\qquad \overline{\mathcal{H}^-} \doteq \{V_+ = 0\} \cap \{U_+ \geq 0\}.
\end{equation}

The vector fields $T$ and $Z$ also extend to $\overline{\mathcal{D}}_{\rm ext}$ where they are given by the following formulas
\begin{equation}\label{Tbifevent}
T = -\kappa_+U_+\partial_{U_+} + \kappa_+V_+\partial_{V_+} - \frac{a}{r_+^2+a^2}\partial_{\varphi_+},\qquad Z = \partial_{\varphi_+}.
\end{equation}

We note, in particular, that for an axisymmetric function $\psi$, i.e., one satisfying $Z\psi = 0$, we have
\begin{equation}\label{formulabifen}
\int_0^{\epsilon}\int_{\mathbb{S}^2}\left(\partial_{U_+}\psi\right)^2|_{\mathcal{H}^-}\, dU_+\, d\mathbb{S}^2 \sim \int_{\log\left(\epsilon\right)\kappa_+^{-1}}^{\infty}\int_{\mathbb{S}^2}\left(\partial_u\psi\right)^2|_{\mathcal{H}^-}e^{\kappa_+u}\, du\, d\mathbb{S}^2.
\end{equation}

We will often refer to the $\{t = 0\}$ hypersurface, defined in terms of Boyer--Lindquist coordinates.
\begin{remark}\label{t0cmpsupp} We emphasise that $\{t = 0\}$ does \underline{not} contain the bifurcation sphere $\mathcal{B}_+$; in particular, if a function has compact support on $\{t = 0\}$ then it must vanish for $r$ sufficiently large and for $r-r_+$ sufficiently small.
\end{remark}

\begin{remark}\label{schw}We have assumed throughout our discussion of the Kerr spacetimes that $a \neq 0$. However, the reader may easily check that the construction of $\overline{\mathcal{D}}_{\rm ext}$ goes through without any difficulty if $a = 0$. 
\end{remark}

The story in the interior is a bit more complicated because we need to work with two separate Kruskal type coordinates. The new set of coordinates $\left(U_-,V_-,\theta,\varphi_-\right) \in (0,\infty) \times (-\infty,0) \times \mathbb{S}^2$ and  $\left(U_{\times},V_{\times},\theta,\varphi_{\times}\right) \in (-\infty,0) \times (0,\infty) \times \mathbb{S}^2$ are defined by
\[U_-\left(t,r\right) \doteq \exp\left(-\kappa_-\left(t-r^*(r)\right)\right),\qquad 
V_-\left(t,r\right) \doteq -\exp\left(\kappa_+\left(t+r^*(r)\right)\right),\]
\[\varphi_-\left(\varphi,t\right) \doteq \varphi - \frac{a}{r_-^2+a^2}t,\]
\[U_{\times}\left(t,r\right) \doteq -\exp\left(-\kappa_+\left(t-r^*(r)\right)\right),\qquad V_{\times}\left(t,r\right) \doteq \exp\left(\kappa_+\left(t+r^*(r)\right)\right),
\]
\[\varphi_{\times}\left(\varphi,t\right) \doteq \varphi - \frac{a}{r_+^2+a^2}t,\]
Now the metric extends to a Lorentzian manifold with stratified boundary $(\overline{\mathcal{D}}_{\rm int},g_{a,M})$ where
\begin{align*}
\overline{\mathcal{D}}_{\rm int} \doteq &\{\left(U_-,V_-,\theta,\varphi_-\right) \in [0,\infty) \times (-\infty,0] \times \mathbb{S}^2\}
\\ \nonumber &\cup \{\left(U_{\times},V_{\times},\theta,\varphi_{\times}\right) \in (-\infty,0]\times [0,\infty) \times \mathbb{S}^2\},
\end{align*}
where we have identified the regions $\{\left(U_-,V_-,\theta,\varphi_-\right) \in (0,\infty) \times (-\infty,0) \times \mathbb{S}^2\}$ and $\{\left(U_{\times},V_{\times},\theta,\varphi_{\times}\right) \in (-\infty,0)\times (0,\infty) \times \mathbb{S}^2\}$ by expressing both $(U_-,V_-,\theta,\varphi_-)$ and $(U_{\times},V_{\times},\theta,\varphi_{\times})$ in Boyer--Lindquist coordinates.

As in the exterior, the various horizons in the interior may be identified with suitable constant $U_-$, $U_{\times}$, $V_-$, or $V_{\times}$ hypersurfaces. In this case the outgoing Cauchy horizon $\mathcal{CH}^+_{\rm out}$ is given by $\{U_- = 0\}\cap \{V_- < 0\}$, the ingoing Cauchy horizon $\mathcal{CH}^+_{\rm in}$ is given by $\{V_- = 0\} \cap \{U_- > 0\}$, the sphere $\{(U_-,V_-) = (0,0)\}$ is denoted by $\mathcal{B}_-$ and is called the ``bifurcation sphere'' (of the Cauchy horizon), and finally, the ingoing future event horizon $\mathcal{H}^+_{\rm in}$ is given by $\{V_{\times} = 0\} \cap \{U_{\times} < 0\}$. We then set
\begin{equation}\label{cauchybif}
\overline{\mathcal{CH}^+_{\rm out}} \doteq \{U_- = 0\}\cap \{V_- \leq 0\},\qquad \overline{\mathcal{CH}^+_{\rm in}} \doteq \{V_- = 0\} \cap \{U_- \geq 0\}.
\end{equation}

As in previous sections, the vector fields $T$ and $Z$ extend to $\overline{\mathcal{D}}_{\rm in}$ where they are given by the following formulas
\begin{equation}\label{Tbifcauchy}
T = -\kappa_-U_-\partial_{U_-} + \kappa_-V_-\partial_{V_-} - \frac{a}{r_-^2+a^2}\partial_{\varphi_-},\qquad Z = \partial_{\varphi_-}.
\end{equation}

We note the following analogue of~\eqref{formulabifen}. Let $\psi$ be an axisymmetric function. We then have
\begin{equation}\label{formulabifen2}
\int_0^{\epsilon}\int_{\mathbb{S}^2}\left(\partial_{V_-}\psi\right)^2|_{\mathcal{CH}^-_{\rm out}}\, dV_-\, d\mathbb{S}^2 \sim \int_{-\log\left(\epsilon\right)\kappa_-^{-1}}^{\infty}\int_{\mathbb{S}^2}\left(\partial_v\psi\right)^2|_{\mathcal{CH}^-_{\rm out}}e^{-\kappa_-v}\, dv\, d\mathbb{S}^2.
\end{equation}

Finally, we set
\[\overline{\mathcal{D}} \doteq \overline{\mathcal{D}}_{\rm ext} \cup \overline{\mathcal{D}}_{\rm int}.\]
This is a smooth Lorentzian manifold with stratified boundary. The boundary $\partial\overline{\mathcal{D}}$ has two components: $\mathcal{H}^+_{\rm in}\cup\mathcal{B}_+\cup\mathcal{H}^-$, which is a smooth manifold, and $\mathcal{CH}^+_{\rm out} \cup \mathcal{B}_- \cup \mathcal{CH}^+_{\rm in}$ which is not smooth at $\mathcal{B}_-$.   

The Penrose diagram of $\overline{\mathcal{D}}$ is given by the following:
\begin{center}
\begin{tikzpicture}
\fill[lightgray] (-2,2)--(2,-2)--(0,-4) -- (-4,0); 
\draw[dashed] (0,0) -- (2,-2) --  (0,-4); 
\draw (0,-4) -- (-2,-2) node[sloped,below,midway]{$\mathcal{H}^-$}; 
\draw (-2,-2) -- (-4,0) node[sloped,below,midway]{$\mathcal{H}^+_{\rm in}$}; 
\draw (-2,-2) -- (0,0) node[sloped,above,midway]{$\mathcal{H}^+_{\rm out}$}; 
\draw (-4,0) -- (-2,2) node[sloped,above,midway]{$\mathcal{CH}^+_{\rm out}$}; 
\draw (-2,2) -- (0,0) node[sloped,above,midway]{$\mathcal{CH}^+_{\rm in}$}; 
\path [draw=black,fill=white] (0,0) circle (1/16); 
\path [draw=black,fill=white] (2,-2) circle (1/16); 
\path [draw=black,fill=white] (0,-4) circle (1/16); 
\path [draw=black,fill=black] (-2,-2) circle (1/16) node[left]{$\mathcal{B}_+$}; 
\path [draw=black,fill=white] (-4,0) circle (1/16); 
\path [draw=black,fill=black] (-2,2) circle (1/16) node[left]{$\mathcal{B}_-$}; 

\node [align = flush center] at (-1,-4.5) {Figure 7: The region $\overline{\mathcal{D}}$};
\end{tikzpicture}
\end{center}

\subsection{Null infinity}\label{pennullsec}

We close this section with a brief discussion of future and past null infinity. We define 
\[\mathcal{I}^+ \doteq \left\{\left(u,\theta,\phi_*\right) \in \mathbb{R} \times \mathbb{S}^2\right\},\qquad \mathcal{I}^- \doteq \left\{\left(v,\theta,\phi^*\right) \in \mathbb{R} \times \mathbb{S}^2\right\}.\]

Future and past null infinity $\mathcal{I}^+$ and $\mathcal{I}^-$ may be attached as a boundary to the spacetime (cf.~\cite{scatter}). However, we will not need a formal definition of this; in this paper the relationship between $\mathcal{I}^{\pm}$ and $\overline{\mathcal{D}}$ will simply by determined by the formulas~\eqref{iplusrad} and~\eqref{iminusrad}.

We depict these boundaries in the Penrose diagram below:

\bigskip

\begin{center}
\begin{tikzpicture}
\fill[lightgray] (0,0)--(2,-2)--(0,-4) -- (-2,-2); 
\draw[dashed] (0,0) -- node[sloped,above,midway]{$\mathcal{I}^+$}(2,-2) -- node[sloped,below,midway]{$\mathcal{I}^-$} (0,-4); 
\draw (0,-4) -- (-2,-2) node[sloped,below,midway]{$\mathcal{H}^-$}; 
\draw (-2,-2) -- (0,0) node[sloped,above,midway]{$\mathcal{H}^+_{\rm out}$}; 
\path [draw=black,fill=white] (0,0) circle (1/16); 
\path [draw=black,fill=white] (2,-2) circle (1/16); 
\path [draw=black,fill=white] (0,-4) circle (1/16); 
\path [draw=black,fill=black] (-2,-2) circle (1/16) node[left]{$\mathcal{B}_+$}; 

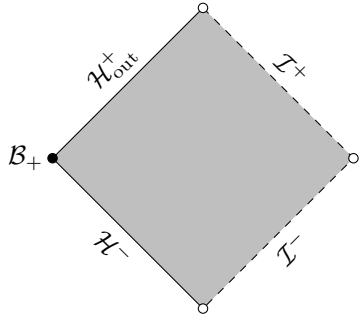
\node [align = flush center] at (0,-4.5) {Figure 8: The region $\overline{\mathcal{D}}_{\rm ext}$ with null infinity $\mathcal{I}^{\pm}$};
\end{tikzpicture}
\end{center}

\section{Basic facts about the wave equation}\label{basic}

In this section we will quickly review some basic and general facts about the wave equation. 

Let $\psi$ be a solution to the wave equation~\eqref{wave} on a Lorentzian manfiold $(\mathcal{M},g)$. We begin by recalling the ``energy-momentum tensor'':
\[\mathbf{T}_{\mu\nu}\left[\psi\right] \doteq \left(\partial_{\mu}\psi\right)\left(\partial_{\nu}\psi\right) - \frac{1}{2}g_{\mu\nu}g^{\gamma\delta}\left(\partial_{\gamma}\psi\right)\left(\partial_{\delta}\psi\right).\]
When there is no risk of confusion we will drop the ``$[\psi]$'' from the definition of $\mathbf{T}_{\mu\nu}$.

For any vector field $X$ we define the associated current:
\[\mathbf{J}^X_{\mu}\left[\psi\right] \doteq \mathbf{T}_{\mu\nu}X^{\nu}.\]
As above, we will drop the ``$[\psi]$'' if there is no risk of confusion. The $\mathbf{J}^X_{\mu}$ current satisfies the following divergence identity:
\[\nabla^{\mu}\mathbf{J}^X_{\mu} = \frac{1}{2}\mathbf{T}^{\mu\nu}\left(\mathcal{L}_Xg\right)_{\mu\nu}.\]

Applying the divergence theorem in a region $\mathcal{B}$ bounded by two homologous hypersurfaces $\Sigma_1$ and $\Sigma_2$ yields the following identity:
\begin{equation}\label{enid}
\int_{\Sigma_1}\mathbf{J}^X_{\mu}n^{\mu}_{\Sigma_1} - \int_{\Sigma_2}\mathbf{J}^X_{\mu}n^{\mu}_{\Sigma_2} =  \frac{1}{2}\int_{\mathcal{B}}\mathbf{T}^{\mu\nu}\left(\mathcal{L}_Xg\right)_{\mu\nu}.
\end{equation}
The normal vectors $n^{\mu}_{\Sigma_1}$ and $n^{\mu}_{\Sigma_2}$ and the associated volume forms are the ones induced by the divergence theorem. In particular, if $X$ is Killing then $\mathcal{L}_Xg$ vanishes and we obtain a conservation law.

If $X$ is timelike and $\Sigma$ is spacelike, then $\mathbf{J}^X_{\mu}n^{\mu}_{\Sigma}$ is positive definite and, in fact, coercive.  If $X$ is timelike and $\Sigma$ is null, then $\mathbf{J}^X_{\mu}n^{\mu}_{\Sigma}$ controls all \emph{tangential} derivatives of the solution. We refer to these two situations as ``non-degenerate energies''. If $X$ is null and $\Sigma$ is spacelike or null, then $\mathbf{J}^X_{\mu}n^{\mu}_{\Sigma}$ is still non-negative, but control over certain derivatives is lost. We refer to these cases as ``degenerate energies''.

See the book~\cite{christ2} for a thorough discussion of energy currents.

Systematic use of the identity~\eqref{enid} allows one to prove general well-posedness and domain of dependence results for the Cauchy and characteristic initial value problem.  In this paper, we will specifically use the well-posedness of the following characteristic initial value problem.

\begin{proposition}\label{totheinterior} Let $\uppsi_{\mathcal{H}^+} : \mathcal{H}^+ \to \mathbb{R}$ be any smooth function which vanishes for $v$ sufficiently negative. Then there exists a unique smooth solution $\psi : \mathcal{D}_{\rm int} \to \mathbb{R}$ to the wave equation~\eqref{wave} such that $\psi|_{\mathcal{H}^+_{\rm in}} = 0$ and $\psi|_{\mathcal{H}^+_{\rm out}} = \uppsi_{\mathcal{H}^+}$.
\end{proposition}

\section{Scattering theory in the exterior}\label{reviewScat}
In the work~\cite{scatter}, in collaboration with Rodnianski, we gave a complete scattering theory for the wave equation~\eqref{wave} on sub-extremal Kerr black hole exteriors; in particular, bounded isomorphisms were established between Hilbert spaces of radiations fields along $\mathcal{H}^+$ and $\mathcal{I}^+$ and $\mathcal{H}^-$ and $\mathcal{I}^-$, and asymptotic completeness was proven. The Hilbert space norms used were \emph{degenerate} along the horizons $\mathcal{H}^{\pm}$ where they were given by the flux of the Hawking vector field~\eqref{hawkvect}.  

In this section we will quickly review the results therein which are relevant for the proof of our main results.

\subsection{Existence and boundedness of radiation fields}
We begin by recalling the definitions of various ``radiation fields'' for smooth solutions.
\begin{definition}\label{defrad}Let $(a,M)$ satisfy $0 \leq |a| < M$ and $\psi : \overline{\mathcal{D}}_{\rm ext} \to \mathbb{R}$ be a smooth solution to the wave equation~\eqref{wave} with compactly supported Cauchy data along $\{t = 0\}$ (cf.\ Remark~\ref{t0cmpsupp}). Then
\begin{enumerate}[I.]
 \item The radiation field $\uppsi_{\mathcal{H}^+}\left(v,\theta,\varphi^*\right) : \mathcal{H}^+ \to \mathbb{R}$  is defined by \
\begin{equation}\label{hplusrad}
\uppsi_{\mathcal{H}^+}\left(v,\theta,\varphi^*\right) \doteq \psi|_{\mathcal{H}^+}\left(v,\theta,\varphi^*\right).
\end{equation}
\item The radiation field $\uppsi_{\mathcal{H}^-}\left(u,\theta,\varphi_*\right) : \mathcal{H}^- \to \mathbb{R}$  is defined by
\begin{equation}\label{hminusrad}
\uppsi_{\mathcal{H}^-}\left(u,\theta,\varphi_*\right) \doteq \psi|_{\mathcal{H}^-}\left(u,\theta,\varphi_*\right).
\end{equation}
\item The radiation field $\upphi_{\mathcal{I}^+}\left(u,\theta,\varphi_*\right) : \mathcal{I}^+ \to \mathbb{R}$ is defined by
\begin{equation}\label{iplusrad}
\upphi_{\mathcal{I}^+}\left(u,\theta,\varphi_*\right) \doteq \lim_{r\to\infty}\left(r\psi\right)\left(u,r,\theta,\varphi_*\right).
\end{equation}
\item The radiation field $\upphi_{\mathcal{I}^-}\left(v,\theta,\varphi^*\right) : \mathcal{I}^- \to \mathbb{R}$ is defined by
\begin{equation}\label{iminusrad}
\upphi_{\mathcal{I}^-}\left(v,\theta,\varphi^*\right) \doteq \lim_{r\to\infty}\left(r\psi\right)\left(v,r,\theta,\varphi^*\right).
\end{equation}
\end{enumerate}
It follows from the analysis in Section 4 of~\cite{scatter} that these limits are all well-defined and smooth.
\end{definition}

This next theorem shows that the radiation fields from Definition~\ref{defrad} can be defined for a more general class of solutions and can be controlled globally in terms of Cauchy data.  For the applications to this paper, it suffices to work in the class of axisymmetric solutions, i.e.,~solutions $\psi$ which satisfy $\partial_{\varphi}\psi = 0$. The restriction to axisymmetry serves only to simplify the exposition of the paper and is in no way essential to our main results. In particular, it is not important for our constructions that the solutions $\psi$ are not amplified by superradiance.
\begin{oldtheorem}\label{radexist}\cite{scatter} (Existence and boundedness of radiation fields) Let $\psi : \mathcal{D}_{\rm ext} \to \mathbb{R}$ be a smooth axisymmetric solution to the wave equation~\eqref{wave} such that the Cauchy data $\left(\psi,n_{\{t=0\}}\psi\right)$ for $\psi$ along the $\{t = 0\}$ hypersurface lie in the closure of compactly supported smooth functions under the norm
\begin{equation}\label{tnorm}\left\vert\left\vert \left(\psi,n_{\{t=0\}}\psi\right)\right\vert\right\vert^2_{\mathcal{E}^T} \doteq \int_{r_+}^{\infty}\int_{\mathbb{S}^2}\left[\frac{r^2}{\Delta}\left(\partial_t\psi\right)^2 + \frac{\Delta}{r^2}\left(\partial_r\psi\right)^2 + r^{-2}\left(\partial_{\theta}\psi\right)^2\right] r^2\, dr\, d\mathbb{S}^2.
\end{equation}
Let $\{\psi^{(i)}\}$ be a sequence of smooth solutions to the wave equation corresponding to a sequence $\left\{\left(\psi^{(i)},n_{\{t=0\}}\psi^{(i)}\right)\right\}_{i=1}^{\infty}$ of smooth Cauchy data which are compactly supported on $\{t = 0\}$ and converge to $\left(\psi,n_{\{t=0\}}\psi\right)$ in $\left\vert\left\vert\cdot\right\vert\right\vert_{\mathcal{E}^T}$.

Then the derivatives of the radiation fields $\partial_u\upphi^{(i)}_{\mathcal{I}^+}$, $\partial_v\upphi^{(i)}_{\mathcal{I}^-}$, $\partial_u\uppsi^{(i)}_{\mathcal{H}^-}$ and $\partial_v\uppsi^{(i)}_{\mathcal{H}^+}$ each converge in $L^2\left(\mathbb{R}\times\mathbb{S}^2\right)$ to functions $\partial_u\upphi_{\mathcal{I}^+}$, $\partial_v\upphi_{\mathcal{I}^-}$, $\partial_u\uppsi_{\mathcal{H}^-}$ and $\partial_v\uppsi_{\mathcal{H}^+}$. Furthermore, these functions $\upphi_{\mathcal{I}^{\pm}}$ and $\uppsi_{\mathcal{H}^{\pm}}$ are independent of the choice of the sequence $\left\{\left(\psi_i,n_{\{t=0\}}\psi_i\right)\right\}_{i=1}^{\infty}$, and we thus define them to be the radiation fields corresponding to $\psi$.

Finally, the radiation fields of $\psi$ are all axisymmetric and are controlled by $\psi$'s Cauchy data in the following sense:
\[\left\vert\left\vert \partial_u\upphi_{\mathcal{I}^+}\right\vert\right\vert_{L^2} + \left\vert\left\vert \partial_v\upphi_{\mathcal{I}^-}\right\vert\right\vert_{L^2} + \left\vert\left\vert \partial_u\uppsi_{\mathcal{H}^-}\right\vert\right\vert_{L^2} + \left\vert\left\vert \partial_v\uppsi_{\mathcal{H}^+}\right\vert\right\vert_{L^2} \leq C\left\vert\left\vert \left(\psi,n_{\{t=0\}}\psi\right)\right\vert\right\vert^2_{\mathcal{E}^T}.\]
\end{oldtheorem}
\begin{proof}Keeping in mind the axisymmetry of $\psi$, we note that the norm~\eqref{tnorm} is simply the $\mathbf{J}^T$ (or $\mathbf{J}^K$) flux of the solution through the hypersurface $\{t = 0\}$.

The theorem then  follows from Theorems 8.2.3.~of~\cite{scatter} (specialised to the case of axisymmetric solutions $\psi$), commuting with $T$ and the Carter operator $Q$ (see~\cite{blueand}), and straightforward elliptic estimates.

\end{proof}
\subsection{Existence of scattering states}
Next, we cite a theorem concerning the existence of solutions with prescribed radiation fields along $\mathcal{I}^+$ or $\mathcal{I}^-$.
\begin{oldtheorem}\label{scatexist}\cite{scatter} (Existence of scattering states with data along $\mathcal{I}^{\pm}$) Let $\upphi_{\mathcal{I}^{\pm}} : \mathcal{I}^{\pm} \to \mathbb{R}$ be an axisymmetric function in $\cap_{s=1}^{\infty}\dot{H}^s\left(\mathbb{R}\times \mathbb{S}^2\right)$.

Then there exists a unique smooth solution $\psi : \mathcal{D}_{\rm ext} \to \mathbb{R}$ to the wave equation~(\ref{wave}) such that
\begin{enumerate}[I.]
    \item The Cauchy data $\left(\psi|_{\{t=0\}},n_{\{t=0\}}\psi|_{\{t=0\}}\right)$ lie in the closure of smooth compactly supported functions under the norm $\left\vert\left\vert \cdot\right\vert\right\vert_{\mathcal{E}^T}$.
    \item The radiation field along $\mathcal{I}^{\pm}$ is given by $\upphi_{\mathcal{I}^{\pm}}$.
    \item The radiation field along $\mathcal{H}^{\pm}$ vanishes.
\end{enumerate}

Furthermore, $\psi$ is axisymmetric, the radiation fields $\uppsi_{\mathcal{H}^{\mp}}$ and $\upphi_{\mathcal{I}^{\mp}}$ lie in $\cap_{s=1}^{\infty}\dot{H}^s\left(\mathbb{R}\times\mathbb{S}^2\right)$, $T\uppsi_{\mathcal{H}^{\mp}}$ does \underline{not} vanish identically, and the radiation fields satisfy
\begin{align}\label{radbound}
\int_{-\infty}^{\infty}\int_{\mathbb{S}^2}&\left(T\uppsi_{\mathcal{H}^{\mp}}\right)^2\, dx\, d\mathbb{S}^2 + \int_{-\infty}^{\infty}\int_{\mathbb{S}^2}\left(T\upphi_{\mathcal{I}^{\mp}}\right)^2\, dx\, d\mathbb{S}^2 \leq C\int_{-\infty}^{\infty}\int_{\mathbb{S}^2}\left(T\upphi_{\mathcal{I}^{\pm}}\right)^2\, dx\, d\mathbb{S}^2,
\end{align}
where the constant $C$ does not depend on $\psi$.

\end{oldtheorem}
\begin{proof}Keeping the axisymmetry of $\psi$ in mind, we note that the terms of the left hand side of~\eqref{radbound} refer to the $\mathbf{J}^T$, or $\mathbf{J}^K$, flux of $\psi$ through the future (past) event horizon $\mathcal{H}^{\pm}$ and future (past) null infinity $\mathcal{I}^{\pm}$ respectively. The term on the right hand side refers to the $\mathbf{J}^T$, or $\mathbf{J}^K$, flux of $\psi$ through past (future) null infinity $\mathcal{I}^{\mp}$. 

The theorem then follows from Theorems 5 and 10 of~\cite{scatter}, commuting with $T$, $Z$ and the Carter operator $Q$ (see~\cite{blueand}), and straightforward elliptic estimates.
\end{proof}

The following diagrams schematically depict the maps $\upphi_{\mathcal{I}^{\pm}}\mapsto \left(\uppsi_{\mathcal{H}^{\mp}},\upphi_{\mathcal{I}^{\mp}}\right)$ given by Theorem~\ref{scatexist}.
\begin{center}
\begin{tikzpicture}
\fill[lightgray] (0,0)--(2,-2)--(0,-4) -- (-2,-2)--(0,0); 
\draw[dashed] (0,0) -- (2,-2) node[sloped,above,midway]{$\mathcal{I}^+$}; 
\draw[dashed] (2,-2) -- (0,-4) node[sloped,below,midway]{$\mathcal{I}^-$}; 
\draw (0,-4) -- (-2,-2) node[sloped,below,midway]{$\mathcal{H}^-$}; 
\draw (-2,-2) -- (0,0) node[sloped,above,midway]{$\mathcal{H}^+_{\rm out}$}; 
\path [draw=black,fill=white] (0,0) circle (1/16); 
\path [draw=black,fill=white] (2,-2) circle (1/16); 
\path [draw=black,fill=white] (0,-4) circle (1/16); 
\path [draw=black,fill=black] (-2,-2) circle (1/16) node[left]{$\mathcal{B}_+$}; 
\draw [very thick] (1,-1) to[out = 190,in = 90] (0,-2);
\draw[very thick,->] (0,-2) to[out = -90,in = 65] (-.9,-2.9);
\draw[very thick,->] (0,-2) to[out = -90,in = 115] (.9,-2.9);

\fill[lightgray] (6,0)--(8,-2)--(6,-4) -- (4,-2)--(6,0); 
\draw[dashed] (6,0) -- (8,-2) node[sloped,above,midway]{$\mathcal{I}^+$}; 
\draw[dashed] (8,-2) -- (6,-4) node[sloped,below,midway]{$\mathcal{I}^-$}; 
\draw (6,-4) -- (4,-2) node[sloped,below,midway]{$\mathcal{H}^-$}; 
\draw (4,-2) -- (6,0) node[sloped,above,midway]{$\mathcal{H}^+_{\rm out}$}; 
\path [draw=black,fill=white] (6,0) circle (1/16); 
\path [draw=black,fill=white] (8,-2) circle (1/16); 
\path [draw=black,fill=white] (6,-4) circle (1/16); 
\path [draw=black,fill=black] (4,-2) circle (1/16) node[left]{$\mathcal{B}_+$}; 
\draw [very thick] (7,-3) to[out = 170,in = -90] (6,-2);
\draw[very thick,->] (6,-2) to[out = 90,in = -65] (5.1,-1.1);
\draw[very thick,->] (6,-2) to[out = 90,in = 245] (6.9,-1.1);
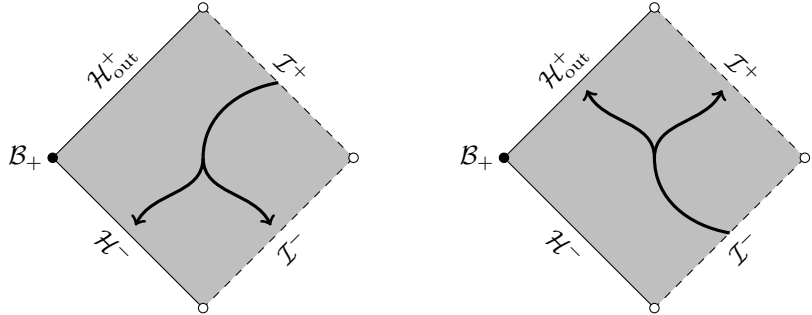
\node [align = flush center] at (3,-4.5) {Figure 9: The scattering maps};
\end{tikzpicture}
\end{center}

When we study the blow-up of solutions along $\mathcal{CH}^+_{\rm in}$ it will be useful to note that if the scattering data vanishes along $\mathcal{H}^-$ and the scattering data along $\mathcal{I}^-$ vanishes in a neighbourhood of past timelike infinity, then we control the full non-degenerate energy flux of the corresponding solution along $\mathcal{H}^+$.
\begin{oldtheorem}\label{betterscatexist} \cite{scatter} (Non-degenerate boundedness for solutions supported away from $\mathcal{H}^-$) Let $\upphi_{\mathcal{I}^-} : \mathcal{I}^- \to \mathbb{R}$ be an axisymmetric function in $\cap_{s=1}^{\infty}\dot{H}^s\left(\mathbb{R}\times \mathbb{S}^2\right)$ which is supported in $\{v \geq 1\}$.

Let $\psi : \mathcal{D}_{\rm ext} \to \mathbb{R}$ be the unique solution to the wave equation such that the radiation field along $\mathcal{I}^-$ is given by $\upphi_{\mathcal{I}^-}$ and such that the radiation field along $\mathcal{H}^-$ vanishes.

Then $\psi$ extends to a smooth solution on $\overline{\mathcal{D}_{\rm ext}}$, and we have the following improved estimate over Theorem~\ref{scatexist}:
\[\int_1^{\infty}\int_{\mathbb{S}^2}\left[\left(\partial_v\uppsi_{\mathcal{H}^+}\right)^2 + \left(\partial_{\theta}\uppsi_{\mathcal{H}^+}\right)^2\right]\, dv\, d\mathbb{S}^2 \leq C\int_1^{\infty}\int_{\mathbb{S}^2}\left(\partial_v\upphi_{\mathcal{I}^-}\right)^2\, dv\, d\mathbb{S}^2,\]
where the constant $C$ does not depend on $\psi$.

\end{oldtheorem}
\begin{proof}This follows immediately from Theorems 1 and 4 of~\cite{scatter} and a straightforward domain of dependence argument.
\end{proof}

\subsection{Time-translation invariance and finite speed of propagation}
This next proposition is the statement that the scattering maps of Theorem~\ref{scatexist} commute with time-translation.
\begin{proposition}\label{timetranslate}Using Theorem~\ref{scatexist}, let
\[\mathscr{B} : \cap_{s=1}^{\infty}\dot{H}^s_{\rm axi}\left(\mathbb{R}\times \mathbb{S}^2\right) \to C_{\rm axi}^{\infty}\left(\mathcal{D}_{\rm ext}\right)\]
denote  the map which takes an axisymmetric function $\upphi_{\mathcal{I}^{\pm}} \in \cap_{s=1}^{\infty}\dot{H}^s\left(\mathbb{R}\times \mathbb{S}^2\right)$ to the unique axisymmetric solution $\psi : \mathcal{D}_{\rm ext} \to \mathbb{R}$ to the wave equation with radiation field $\upphi_{\mathcal{I}^{\pm}}$ along $\mathcal{I}^{\pm}$ and a vanishing radiation field along $\mathcal{H}^{\pm}$.

Next, using Theorem~\ref{scatexist} again, let
\[\mathscr{T} : \cap_{s=1}^{\infty}\dot{H}^s_{\rm axi}\left(\mathbb{R}\times \mathbb{S}^2\right) \to \cap_{s=1}^{\infty}\dot{H}^s_{\rm axi}\left(\mathbb{R}\times \mathbb{S}^2\right),\]
denote the map which takes an axisymmetric function $\upphi_{\mathcal{I}^{\pm}} : \mathcal{I}^{\pm} \to \mathbb{R}$ in $\cap_{s=1}^{\infty}\dot{H}^s\left(\mathbb{R}\times \mathbb{S}^2\right)$ to the corresponding axisymmetric radiation field $\uppsi_{\mathcal{H}^{\mp}}$ along $\mathcal{H}^{\mp}$.

Finally, for every $t_0 \in \mathbb{R}$, we define the corresponding time-translation operators
\[\mathscr{L}_{t_0} : \cap_{s=1}^{\infty}\dot{H}^s_{\rm axi}\left(\mathbb{R}\times \mathbb{S}^2\right) \to \cap_{s=1}^{\infty}\dot{H}^s_{\rm axi}\left(\mathbb{R}\times \mathbb{S}^2\right),\quad \mathscr{Q}_{t_0} : C_{\rm axi}^{\infty}\left(\mathcal{D}_{\rm ext}\right) \to C_{\rm axi}^{\infty}\left(\mathcal{D}_{\rm ext}\right),\]
by the formulas
\[\mathscr{L}_{t_0}f\left(x,\theta,\varphi\right) \doteq f\left(x-t_0,\theta,\varphi\right),\qquad \mathscr{Q}_{t_0}f\left(t,r,\theta,\varphi\right) \doteq f\left(t-t_0,r,\theta,\varphi\right).\]

Then, for every $t_0 \in \mathbb{R}$, the time-translation maps commute with the scattering maps $\mathscr{B}$ and $\mathscr{T}$ in the following sense:
\[\mathscr{Q}_{t_0}\circ \mathscr{B} = \mathscr{B}\circ\mathscr{L}_{t_0},\qquad \mathscr{L}_{t_0}\circ\mathscr{T} = \mathscr{T} \circ \mathscr{L}_{t_0}.\]

\end{proposition}
\begin{proof}This follows immediately from the fact that $T$ is Killing and the uniqueness of the maps from Theorem~\ref{scatexist}.
\end{proof}

This final proposition is a consequence of the finite speed of propagation.
\begin{proposition}\label{finitespeed}Let $\upphi_{\mathcal{I}^+}(u,\theta,\phi_*) : \mathcal{I}^+ \to \mathbb{R}$ be an axisymmetric function satisfying ${\rm supp}\left(\upphi_{\mathcal{I}^+}\right) \subset \{u \leq 0\}$. Then let $\psi : \mathcal{D} \to \mathbb{R}$ be the unique solution to the wave equation such that the radiation field along $\mathcal{I}^+$ is given by $\upphi_{\mathcal{I}^+}$ and such that the radiation field along $\mathcal{H}^+$ vanishes. 

Then there exists $u_{\dagger} \in \mathbb{R}$ such that ${\rm supp}\left(\psi\right)  \subset \{u \in (u_{\dagger},-\infty)\}$.
\end{proposition}
\begin{proof}Pick $u_{\dagger}$ such that $J^+\left(\{u \leq 0\}\right) \subset \{u \in (u_{\dagger},-\infty)\}$. Then the proposition follows immediately from the physical space characterization of the scattering map given in Section 9.1.4 of~\cite{scatter} and a straightforward domain of dependence argument.
\end{proof}
\begin{remark}\label{finitespeedrmk}Analogously, if $\upphi_{\mathcal{I}^-}(v,\theta,\phi^*) : \mathcal{I}^- \to \mathbb{R}$ is supported in $\{v \geq 0\}$, then there exists $v_{\dagger}$ such that the corresponding solution $\psi$ satisfies ${\rm supp}\left(\psi\right) \subset \{v \in (v_{\dagger},\infty)\}$.
\end{remark}

\section{Precise statements of the main results}\label{secResults}
In this section we will present precise formulations of our main results.

\subsection{Blow-up at the event horizon}\label{blowupEHsec}
Our first main result constructs solutions $\psi$ whose radiation field $\uppsi_{\mathcal{H}^+}$ vanishes along $\mathcal{H}^+$ and whose radiation field $\upphi_{\mathcal{I}^+}$ decays at an arbitrarily fast polynomial rate along $\mathcal{I}^+$, but nevertheless, the solution does not lie in $H^1_{\rm loc}$ of any neighbourhood of any point on $\overline{\mathcal{H}^+_{\rm out}}$.

\begin{theorem}\label{IplusToEvent}Let $p \gg 1$ be sufficiently large and $(a,M)$ satisfy $0 \leq |a| < M$. Then there exists a smooth axisymmetric $\upphi_{\mathcal{I}^+}(u,\theta,\varphi_*) : \mathcal{I}^+ \to \mathbb{R}$ supported in $\{u \geq -1\}$ and satisfying
\[\left|T^i\partial_{\theta}^j\upphi_{\mathcal{I}^+}\right| \leq C_{i,j}\left(1+\left|u\right|\right)^{-p}\qquad \forall i,j \geq 0,\]
such that if  $\psi: \mathcal{D}_{\rm ext} \to \mathbb{R}$ is the unique solution to the wave equation such that the radiation field along $\mathcal{I}^+$ is given by $\upphi_{\mathcal{I}^+}$ and such that the radiation field along $\mathcal{H}^+$ vanishes, then 
\[\forall q \in \overline{\mathcal{H}^+_{\rm out}}\text{ and open }U \owns q\ \psi \not\in H^1_{\rm loc}\left(U\right).\]

\end{theorem}

We will prove Theorem~\ref{IplusToEvent} in Section~\ref{secPfThm1}.

\begin{remark}\label{truetail}Though we will not pursue this here, a straightforward modification of the proof of Theorem~\ref{IplusToEvent} would allow for the radiation field $\upphi_{\mathcal{I}^+}$ to be independent of $\theta$ and to satisfy 
\[\left|T^i\upphi_{\mathcal{I}^+}\right| \leq C_i\left(1+\left|u\right|\right)^{-p-i}\qquad \forall i\geq 0.\]
\end{remark}

The following is an immediate corollary.
\begin{corollary}The forward map $\mathscr{F}_+$ from Theorem 1 of~\cite{scatter} which sends finite non-degenerate energy Cauchy data to finite non-degenerate energy radiation fields along $\mathcal{H}^+_{\geq 0}$ and $\mathcal{I}^+$ is \underline{not} surjective.
\end{corollary}
\begin{remark}In Theorem 11.1 of~\cite{scatter} we proved this corollary for the case when $a = 0$.
\end{remark}

\subsection{Blow-up at the Cauchy horizon}\label{blowupCHsec}
One second main result constructs solutions $\psi$ such that $\psi$ vanishes along $\mathcal{H}^- \cup \mathcal{B}_+ \cup \mathcal{H}^+_{\rm in}$ and the radiation field $\upphi_{\mathcal{I}^-}$ along $\mathcal{I}^-$ decays at an arbitrarily fast polynomial rate towards spacelike infinity, but nevertheless, the solution does not lie in $H^1_{\rm loc}$ of any neighbourhood of any point on $\overline{\mathcal{CH}^+_{\rm in}}$.
\begin{theorem}\label{IminusToCauchy}Let $p \gg 1$ be sufficiently large and $(a,M)$ satisfy $0 < |a| < M$ . Then there exists a smooth axisymmetric $\upphi_{\mathcal{I}^-}(v,\theta,\varphi^*) : \mathcal{I}^- \to \mathbb{R}$ supported in $\{v \geq 1\}$ and satisfying
\[\left|T^i\partial_{\theta}^j\upphi_{\mathcal{I}^-}\right| \leq C_{i,j}\left|v\right|^{-p}\qquad \forall i,j \geq 0,\]
such that if $\psi: \mathcal{D}_{\rm ext} \to \mathbb{R}$ denotes the unique solution to the wave equation such that the radiation field along $\mathcal{I}^-$ is given by $\upphi_{\mathcal{I}^-}$ and such that the radiation field along $\mathcal{H}^-\cup\mathcal{B}_+\cup\mathcal{H}^+_{\rm in}$ vanishes, then $\psi$ extends uniquely to $\mathcal{D}$ as a smooth solution to the wave equation and 
\[\forall q \in \overline{\mathcal{CH}^+_{\rm in}}\text{ and open }U \owns q,\ \psi \not\in H^1_{\rm loc}\left(U\right).\]

\end{theorem}

We will prove Theorem~\ref{IminusToCauchy} in Section~\ref{secPfThm2}.

\begin{remark}\label{truetail2} Analogously to Remark~\ref{truetail}, we note that although we will not pursue this here, a straightforward modification of the proof of Theorem~\ref{IminusToCauchy} would allow for the radiation field $\upphi_{\mathcal{I}^-}$ to be independent of $\theta$ and to satisfy 
\[\left|T^i\upphi_{\mathcal{I}^+}\right| \leq C_i\left(1+\left|u\right|\right)^{-p-i}\qquad \forall i\geq 0.\]
\end{remark}

\section{Proof of Theorem~\ref{IplusToEvent}}
\label{secPfThm1}
In this section we will prove Theorem~\ref{IplusToEvent}. In Section~\ref{constructrad} we will construct the radiation field $\upphi_{\mathcal{I}^+}$ in the statement of the theorem, in Section~\ref{almorth} we will establish some ``almost orthogonality properties'' of the transmission map and then, in Section~\ref{prooftheo1}, we will establish the desired blow-up.

\subsection{Construction of the radiation field $\upphi_{\mathcal{I}^+}$}\label{constructrad}
Let $\left(\upphi_0\right)_{\mathcal{I}^+}(u,\theta,\varphi_*) : \mathcal{I}^+ \to \mathbb{R}$ be a non-zero smooth axisymmetric function which is compactly supported in $\{u \in (-1,0)\}$. Then, using Theorem~\ref{scatexist}, let $\psi_0 : \mathcal{D}_{\rm ext} \to \mathbb{R}$ be the unique solution to the wave equation with a vanishing radiation field along $\mathcal{H}^+$ and with radiation field $\left(\upphi_0\right)_{\mathcal{I}^+}$ along $\mathcal{I}^+$.

Recall that Theorem~\ref{scatexist} guarantees that $\partial_u\left(\uppsi_0\right)_{\mathcal{H}^-}$ lies in $L^2\left(\mathbb{R}\times\mathbb{S}^2\right)$ and does not vanish identically. In particular, after rescaling $\left(\upphi_0\right)_{\mathcal{I}^+}$ and appealing to Proposition~\ref{finitespeed}, we can assume without loss of generality that
\begin{equation}\label{normalize}
\int_{-\infty}^{u_{\dagger}}\int_{\mathbb{S}^2}\left(\partial_u\left(\uppsi_0\right)_{\mathcal{H}^-}\right)^2\, du\, d\mathbb{S}^2 = 1.
\end{equation}

Next, we let $\epsilon > 0$ be a sufficiently small constant, and  we inductively construct a monotonic sequence $u_i\to-\infty$ such that

\[u_i < 0,\]
\[\int_{-\infty}^{u_i}\int_{\mathbb{S}^2}\left(\partial_u\left(\uppsi_0\right)_{\mathcal{H}^-}\right)^2\, du\, d\mathbb{S}^2 \leq \epsilon^22^{-4i}\int_{u_i/2}^{\infty}\int_{\mathbb{S}^2}\left(\partial_u\left(\psi_0\right)_{\mathcal{H}^-}\right)^2\, du\, d\mathbb{S}^2,\]
\begin{equation}\label{ujbig}
|u_i|/2 \geq |u_1|+\cdots+|u_{i-1}| + \left|u_{\dagger}\right|,
\end{equation}
where $u_{\dagger}$ is the constant from Proposition~\ref{finitespeed}.

(Of course, the sequence is not uniquely defined.)

We then set
\[\tilde u_i \doteq u_1 + \cdots + u_i,\]
and, recalling the number $p$ from Theorem~\ref{IplusToEvent}, define
\[\left(\upphi_i\right)_{\mathcal{I}^+}\left(u,\theta,\varphi_*\right) \doteq \left|\tilde u_i\right|^{-p}\left(\upphi_0\right)_{\mathcal{I}^+}\left(u + \tilde u_i,\theta,\varphi\right),\qquad \upphi_{\mathcal{I}^+} \doteq \sum_{i=0}^{\infty}\left(\upphi_i\right)_{\mathcal{I}^+}.\]
Note in particular that $\upphi_{\mathcal{I}^+}$ satisfies the requirements from Theorem~\ref{IplusToEvent}, i.e.,~it is axisymmetric, smooth, supported in $\{u \geq -1\}$ and we may easily check that
\[\left|T^i\partial_{\theta}^j\upphi_{\mathcal{I}^+}\right| \leq C_{i,j}\left(1+\left|u\right|\right)^{-p}\qquad \forall i,j \geq 0.\]

Finally, using Theorem~\ref{scatexist}, we let $\psi_i:\mathcal{D}_{\rm ext} \to \mathbb{R}$ and $\psi : \mathcal{D}_{\rm ext} \to \mathbb{R}$ denote the unique solutions to the wave equation with vanishing radiation fields along $\mathcal{H}^+$ and radiation fields along $\mathcal{I}^+$ given by $\left(\upphi_i\right)_{\mathcal{I}^+}$ and $\upphi_{\mathcal{I}^+}$ respectively. (Note that it is clear that $\upphi_{\mathcal{I}^+}$ satisfies the requirements from Theorem~\ref{scatexist}; in particular, the radiation fields of $\psi$ satisfy the bounds~\eqref{radbound}.)

We close the section by noting the following fundamental relation:
\begin{lemma}\label{thewholepoint}We have
\begin{equation*}
\left(\uppsi_i\right)_{\mathcal{H}^-}\left(u,\theta,\varphi_*\right) = \left|\tilde u_i\right|^{-p}\left(\uppsi_0\right)_{\mathcal{H}^-}\left(u+\tilde u_i,\theta,\varphi\right).
\end{equation*}
\end{lemma}
\begin{proof}This follows immediately from Proposition~\ref{timetranslate} and linearity.
\end{proof}

\subsection{Almost orthogonality}\label{almorth}
Note that a $T$-energy estimate is easily seen to imply the global preservation of the orthogonality of $\left(\upphi_i\right)_{\mathcal{I}^+}$ and $\left(\upphi_j\right)_{\mathcal{I}^+}$ (cf.~Theorem 8 of~\cite{scatter} which shows that the map $\upphi_{\mathcal{I}^+} \mapsto \left(\uppsi_{\mathcal{H}^-},\upphi_{\mathcal{I}^-}\right)$ is pseudo-unitary with respect to the $\mathbf{J}^T$-energy):
\[\int_{-\infty}^{\infty}\int_{\mathbb{S}^2}\left(\partial_u\left(\uppsi_i\right)_{\mathcal{H}^-}\right)\left(\partial_u\left(\uppsi_j\right)_{\mathcal{H}^-}\right)\, du\, d\mathbb{S}^2 + \int_{-\infty}^{\infty}\int_{\mathbb{S}^2}\left(\partial_v\left(\upphi_i\right)_{\mathcal{I}^-}\right)\left(\partial_v\left(\upphi_j\right)_{\mathcal{I}^-}\right)\, du\, d\mathbb{S}^2 = 0.\]
The following estimate, however, quantifies the statement that if $i \neq j$, then the radiation fields of $\psi_i$ and $\psi_j$ restricted to $\mathcal{H}^-$ will still be ``almost orthogonal.''

Since the radiation field along $\mathcal{H}^+$ of each $\psi_i$ and of $\psi$ vanishes, we will not need to refer to the radiation fields along $\mathcal{H}^+$. Thus, in order to simplify the notation, we will denote the radiation field of each $\psi_i$ and of $\psi$ along $\mathcal{H}^-$ by $\uppsi_i$ and $\uppsi$ respectively, instead of $\left(\uppsi_i\right)_{\mathcal{H}^-}$ and $\uppsi_{\mathcal{H}^-}$.

\begin{lemma}\label{cruxofthematter}Let $j > i$ and $s > 0$. Then the following bound holds along $\mathcal{H}^-$.
\begin{equation}\label{almostortho}
\int_s^{\infty}\int_{\mathbb{S}^2}1_{{\rm supp}\left(\partial_u\uppsi_i\right)}\left(\partial_u\uppsi_j\right)^2\, du\, d\mathbb{S}^2 \leq \epsilon^2 2^{-2i-2j}\int_s^{\infty}\int_{\mathbb{S}^2}\left(\partial_u\uppsi_j\right)^2\, du\, d\mathbb{S}^2.
\end{equation}

\end{lemma}
\begin{proof}In the following calculations we will suppress the $\mathbb{S}^2$ volume forms and the dependence of our functions on $\theta$ and $\varphi_*$. We have
\begin{align}
 \nonumber &\int_s^{\infty}\int_{\mathbb{S}^2}1_{{\rm supp}\left(\partial_u\uppsi_i\right)}\left(\partial_u\uppsi_j\right)^2\, du
\\ \nonumber &= \int_s^{\infty}\int_{\mathbb{S}^2}1_{{\rm supp}\left(\left|\tilde u_i\right|^{-p}\partial_u\uppsi_0\left(u+\tilde u_i\right)\right)}\left(\left|\tilde u_j\right|^{-p}\partial_u\uppsi_0\left(u+\tilde u_j\right)\right)^2\, du
\\ \nonumber &= \int_{s+\tilde u_i}^{\infty}\int_{\mathbb{S}^2}1_{{\rm supp}\left(\partial_u\uppsi_0(u)\right)}\left(\left|\tilde u_j\right|^{-p}\partial_u\uppsi_0\left(u+\tilde u_j-\tilde u_i\right)\right)^2\, du
\\  &\label{couldbezero}\leq \int_{s+\tilde u_i}^{u_{\dagger}}\int_{\mathbb{S}^2}\left(\left|\tilde u_j\right|^{-p}\partial_u\uppsi_0\left(u+\tilde u_j-\tilde u_i\right)\right)^2\, du.
\end{align}

Observe that the right hand side of~\eqref{couldbezero} will vanish unless
\begin{equation}\label{sbound}
s \leq \left|\tilde u_i\right| + \left|u_{\dagger}\right|.
\end{equation}
If the right hand side of~\eqref{couldbezero} vanishes, then~\eqref{almostortho} clearly holds. Thus, without loss of generality, we may assume that~\eqref{sbound} holds.

We now continue with our estimate of~\eqref{couldbezero}:

\begin{align*}
 &\int_{s+\tilde u_i}^{u_{\dagger}}\int_{\mathbb{S}^2}\left(\left|\tilde u_j\right|^{-p}\partial_u\uppsi_0\left(u+\tilde u_j-\tilde u_i\right)\right)^2\, du
\\ \nonumber &\qquad \qquad \qquad = \int_{s+\tilde u_j}^{u_{\dagger} + \tilde u_j - \tilde u_i}\int_{\mathbb{S}^2}\left(\left|\tilde u_j\right|^{-p}\partial_u\uppsi_0\left(u\right)\right)^2\, du
\\ \nonumber &\qquad \qquad \qquad \leq \int_{-\infty}^{u_j}\int_{\mathbb{S}^2}\left(\left|\tilde u_j\right|^{-p}\partial_u\uppsi_0\left(u\right)\right)^2\, du
\\ \nonumber &\qquad \qquad \qquad \leq \epsilon^2 2^{-4j}\int_{u_j/2}^{\infty}\int_{\mathbb{S}^2}\left(\left|\tilde u_j\right|^{-p}\partial_u\uppsi_0\left(u\right)\right)^2\, du
\\ \nonumber &\qquad \qquad \qquad = \epsilon^2 2^{-4j}\int_{-u_j/2 - \tilde u_{j-1}}^{\infty}\int_{\mathbb{S}^2}\left(\partial_u\uppsi_j(u)\right)^2\, du
\\ \nonumber &\qquad \qquad \qquad \leq \epsilon^2 2^{-2i-2j}\int_s^{\infty}\int_{\mathbb{S}^2}\left(\partial_u\uppsi_j(u)\right)^2\, du.
\end{align*}
In the last line we have used that $i < j$,~\eqref{sbound} and~\eqref{ujbig}.

\end{proof}

\subsection{Blow-up at the bifurcation sphere}\label{prooftheo1}We continue to employ the convention of the previous section and denote the radiation field of each $\psi_i$ and of $\psi$ along $\mathcal{H}^-$ by $\uppsi_i$ and $\uppsi$ respectively.

In the next proposition we show that $\left(\partial_u\uppsi\right)^2u^{2p}$ is not integrable; the proof is a relatively direct consequence of the construction of $\uppsi$ and the ``almost orthogonality'' of the $\uppsi_i$ established in the previous lemma.

\begin{proposition}There exists a constant $c$ and sequence $\{s_i\}_{i=1}^{\infty}$ with $s_i \to \infty$ such that
\begin{equation}\label{whatagreatlowerbound}
\int_{s_i}^{\infty}\int_{\mathbb{S}^2}\left(\partial_u\uppsi\right)^2\, du\, d\mathbb{S}^2 \geq \frac{c}{(s_i)^{2p}}.
\end{equation}
Consequently,
\begin{equation}\label{itblowsup}
\int_1^{\infty}\int_{\mathbb{S}^2}\left(\partial_u\uppsi\right)^2u^{2p}\, du\, d\mathbb{S}^2 = \infty.
\end{equation}
\end{proposition}
\begin{proof}

Let $s > 1$. We have
\begin{align}\label{lowerbound}
 \int_s^{\infty}\int_{\mathbb{S}^2}\left(\partial_u\uppsi\right)^2\, du\, d\mathbb{S}^2 \geq \sum_{i=0}^{\infty}&\int_s^{\infty}\int_{\mathbb{S}^2}\left(\partial_u\uppsi_i\right)^2\, du\, d\mathbb{S}^2 \\ \nonumber &- 2\sum_{i=0}^{\infty}\sum_{j=i+1}^{\infty}\int_s^{\infty}\int_{\mathbb{S}^2}\left|\partial_u\uppsi_i\partial_u\uppsi_j\right|\, du\, d\mathbb{S}^2.
\end{align}

Lemma~\ref{cruxofthematter} yields
\begin{align}\label{crossterm}
\sum_{i=0}^{\infty}\sum_{j=i+1}^{\infty}&\int_s^{\infty}\int_{\mathbb{S}^2}\left|\partial_u\uppsi_i\partial_u\uppsi_j\right|\, du\, d\mathbb{S}^2
\\ \nonumber &\leq \sum_{i=0}^{\infty}\sum_{j=i+1}^{\infty}\sqrt{\int_s^{\infty}\int_{\mathbb{S}^2}\left(\partial_u\uppsi_i\right)^2\, du\,d\mathbb{S}^2}\sqrt{\int_s^{\infty}\int_{\mathbb{S}^2}1_{{\rm supp}\left(\partial_u\uppsi_i\right)}\left(\partial_u\uppsi_j\right)^2\, du\,d\mathbb{S}^2}
\\ \nonumber &\leq \epsilon\sum_{i=0}^{\infty}\sum_{j=i+1}^{\infty}2^{-i-j}\sqrt{\int_s^{\infty}\int_{\mathbb{S}^2}\left(\partial_u\uppsi_i\right)^2\, du\,d\mathbb{S}^2}\sqrt{\int_s^{\infty}\int_{\mathbb{S}^2}\left(\partial_u\uppsi_j\right)^2\, du\,d\mathbb{S}^2}
\\ \nonumber &\leq 2\epsilon\sum_{i=0}^{\infty}\int_s^{\infty}\int_{\mathbb{S}^2}\left(\partial_u\uppsi_i\right)^2\, du\, d\mathbb{S}^2
\end{align}

Taking $\epsilon < 1/8$ and combining~\eqref{crossterm} with~\eqref{lowerbound} yields
\begin{equation}\label{lowerbound2}
 \int_s^{\infty}\int_{\mathbb{S}^2}\left(\partial_u\uppsi\right)^2\, du\, d\mathbb{S}^2 \geq \frac{1}{2}\sum_{i=0}^{\infty}\int_s^{\infty}\int_{\mathbb{S}^2}\left(\partial_u\uppsi_i\right)^2\, du\, d\mathbb{S}^2.
\end{equation}

Next, using the normalisation~\eqref{normalize}, we may find a constant $c < 0$ such that along $\mathcal{H}^-$
\[\int_c^{u_{\dagger}}\int_{\mathbb{S}^2}\left(\partial_u\uppsi_0\right)^2\, du\, d\mathbb{S}^2 \geq \frac{1}{2}.\]
In particular, it then immediately follows from~\eqref{lowerbound2} and Lemma~\ref{thewholepoint} that for each $\tilde u_i$ we have
\[\int_{\left|\tilde u_i\right| + c}^{\infty}\int_{\mathbb{S}^2}\left(\partial_u\uppsi\right)^2\, du\, d\mathbb{S}^2 \geq \frac{1}{4}\left|\tilde u_i\right|^{-2p}.\]
Since $\lim_{i\to\infty}\left|\tilde u_i\right|\to \infty$, we immediately conclude that~\eqref{whatagreatlowerbound} holds.

In particular,
\[\liminf_{s\to\infty}\int_s^{\infty}\left(\partial_u\uppsi\right)^2u^{2p} \neq 0.\]
This yields~\eqref{itblowsup}.
\end{proof}

Since the local energy of $\psi$ along $\mathcal{H}^-$ in a neighborhood of the bifurcation sphere is proportional to 
\[\int_0^{\epsilon}\int_{\mathbb{S}^2}\left[\left(\partial_U\psi\right)^2 + \left(\partial_{\theta}\psi\right)^2\right]\, dU\, d\mathbb{S}^2,\] 
the lower bound~\eqref{whatagreatlowerbound} and the change of variables formula~\eqref{formulabifen} already immediately imply that the local energy of $\psi$ is infinite near the bifurcation sphere.

\subsection{Propagation of singularities}\label{notstandardprop}
One approach to finishing the proof of Theorem~\ref{IplusToEvent} is to apply a propagation of singularities type argument to show that the local energy blow-up from the previous section is immediately inherited along all of $\mathcal{H}^+$. (This is in fact precisely what we
shall do in the case of Theorem~\ref{IminusToCauchy}; see Section~\ref{propsing} below.) However, for technical reasons it will be easier here to directly show that the lower bound~\eqref{whatagreatlowerbound} is propagated along suitable spacelike hypersurfaces intersecting $\mathcal{H}^+$.

We define two families of hypersurfaces parametrized by $\tau \in \mathbb{R}$:
\[\Sigma_{\tau} \doteq \left\{(t,r,\theta,\varphi) : t = \tau - r^*\left(r\right)+r\right\},\qquad \hat{\Sigma}_{\tau} \doteq \left\{(t,r,\theta,\varphi) : t = \tau + r^*\left(r\right)-r\right\}.\]

One may easily check that for $r$ sufficiently close to $r_+$, each $\Sigma_{\tau}$ is spacelike and smoothly extends to $\overline{\mathcal{D}_{\rm ext}}$ where it intersects the future event horizon $\mathcal{H}^+$ transversally, and that, in fact, $\{\Sigma_{\tau}\cap [r_+,r_0]\}_{\tau \in \mathbb{R}}$ foliates $\left(\overline{\mathcal{D}_{\rm ext}}\cap \{r \in [r_+,r_0]\}\right)\setminus \overline{\mathcal{H}^-}$ for $r_0$ sufficiently close to $r_+$. Similarly, for $r$ sufficiently close to $r_+$, each $\hat{\Sigma}_{\tau}$ is spacelike and smoothly extends to $\overline{\mathcal{D}_{\rm ext}}$ where it intersects the past event horizon $\mathcal{H}^-$ transversally, and $\{\hat{\Sigma}_{\tau}\cap [r_+,r_0]\}_{\tau \in \mathbb{R}}$ foliates $\left(\overline{\mathcal{D}_{\rm ext}}\cap \{r \in [r_+,r_0]\}\right)\setminus \overline{\mathcal{H}^+}$ for $r_0$ sufficiently close to $r_+$.

Let $Y$ denote the $\partial_r$ vector field in $(v,r,\theta,\varphi^*)$ coordinates. In order to finish the proof of Theorem~\ref{IplusToEvent} it suffices to show the stronger statement that $Y\psi \not\in L^2\left(\Sigma_{\tau}\right)$ for every $\tau \in \mathbb{R}$. Pick and fix some $\tau_0 \in \mathbb{R}$.  First of all, for all $r_0 \in (r_+,R]$, a straightforward calculations yield the following two relations:
\begin{equation}\label{whatsh1}
\left\vert\left\vert Y\psi\right\vert\right\vert^2_{L^2\left(\Sigma_{\tau_0}\cap [r_+,r_0]\right)} \sim_R \int_{r_+}^{r_0}\int_{\mathbb{S}^2}\left(Y\psi\right)^2\big|_{\Sigma_{\tau_0}}\, dr\, d\mathbb{S}^2,
\end{equation}
\begin{equation}\label{whatisjk}
\int_{\Sigma_{\tau_0}\cap [r_+,r_0]}\mathbf{J}^K_{\mu}\left[\psi\right]n^{\mu}_{\Sigma_{\tau}} \sim_R \int_{r_+}^{r_0}\int_{\mathbb{S}^2}\left[\left(\partial_v\psi\right)^2 + \left(r-r_+\right)\left(Y\psi\right)^2 + \left(\partial_{\theta}\psi\right)^2\right]\big|_{\Sigma_{\tau_0}}\, dr\, d\mathbb{S}^2,
\end{equation}
where $K$  denotes the Hawking vector field~\eqref{hawkvect}. We recall that $K$ is a Killing vector field and that $K$ is timelike for $r_+ < r < r_0$ for $r_0$ sufficiently close to $r_+$. Lastly, we emphasize that the constant in the $\sim_R$ depends only on $R$. 

Now, observe that the hypersurfaces $\Sigma_{\tau_0}$ and $\hat{\Sigma}_s$ intersect where $2\left(r^*-r\right) = -s+\tau$. In particular, a straightforward density argument, the fact that $\uppsi_{\mathcal{H}^+} = 0$, $\mathbf{J}^K$-energy estimates in the region $J^+\left(\hat{\Sigma}_s\right) \cap J^-\left(\Sigma_{\tau_0}\right) \cap J^-\left(\mathcal{H}^+_{\rm out}\right) \cap J^+\left(\mathcal{H}^-\right)$ (which we have indicated in the diagram below\footnote{We have already noted that the Penrose diagram refers to the global domain of a double-null foliation. However, we will abuse the notation and use the Penrose diagram to depict the hypersurfaces $\Sigma_{\tau_0}$ and $\hat{\Sigma}_s$.}) and~\eqref{whatagreatlowerbound}, immediately yield
\begin{equation}\label{thisiswhatwasyielded}
\int_{r_+}^{r_++ 2e^{-\kappa_+s_i}}\int_{\mathbb{S}^2}\left(r-r_+\right)\left(Y\psi\right)^2\big|_{\Sigma_{\tau_0}}\, dr\, d\mathbb{S}^2 \geq \frac{c}{(s_i)^{2p}},
\end{equation}
for a sequence $\{s_i\}$ with $s_i\to\infty$ as $i\to\infty$, and a constant $c$ which is independent of $i$.

\begin{center}
\begin{tikzpicture}
\fill[lightgray] (0,0)--(2,-2)--(0,-4) -- (-2,-2)--(0,0); 
\fill[gray] (-1.25,-1.25) to[out = -10, in =150] (.1,-1.6) to[out =200, in = 10] (-1.7,-2.3) -- (-2,-2) -- (-1.25,-1.25);
\draw[dashed] (0,0) -- (2,-2) --  (0,-4); 
\draw (0,-4) -- (-2,-2) node[sloped,below,midway]{$\mathcal{H}^-$}; 
\draw (-2,-2) -- (0,0) node[sloped,above,midway]{$\mathcal{H}^+_{\rm out}$}; 
\path [draw=black,fill=white] (0,0) circle (1/16); 
\path [draw=black,fill=white] (2,-2) circle (1/16); 
\path [draw=black,fill=white] (0,-4) circle (1/16); 
\path [draw=black,fill=black] (-2,-2) circle (1/16) node[below]{$\mathcal{B}_+$}; 

\draw (-1.25,-1.25) to[out = -10, in =150] node[above, midway, sloped,scale = .8]{$\Sigma_{\tau_0}$} (.1,-1.6);
\draw (-1.7,-2.3) to[out =10, in = 200] node[below, midway, sloped,scale = .8]{$\hat{\Sigma}_s$} (.1,-1.6);
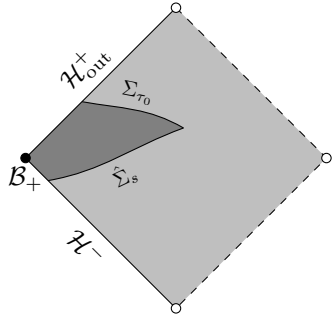
\node [align = flush center] at (0,-4.5) {Figure 10: The region of the $\mathbf{J}^K$-energy estimate};
\end{tikzpicture}
\end{center}

Combining~\eqref{thisiswhatwasyielded} with~\eqref{whatsh1} immediately implies that $Y\psi \not\in L^2\left(\Sigma_{\tau_0}\cap [r_+,r_0]\right)$ and finishes the proof.

\section{Interior scattering and the proof of Theorem~\ref{IminusToCauchy}}\label{secPfThm2}
We now turn to the construction of solutions which blow-up along $\mathcal{CH}^+_{\rm in}$. First, in Section~\ref{radcauchy} we establish the necessary scattering theory results about transmission from $\mathcal{I}^-$ to $\mathcal{CH}^+_{\rm out}$. 
We then give the proof of Theorem~\ref{theorem2} in Section~\ref{proofofthe2}.

\subsection{Scattering in the black hole interior}\label{radcauchy}
In this section we will extend the most elementary scattering theory statements to the black hole interior; in particular, we will show that the transmission map from $\mathcal{I}^-$ to $\mathcal{CH}^+_{\rm out}$ is well defined and non-vanishing:

\begin{oldtheorem}\label{tothecauchyandbeyond}Let $\upphi_{\mathcal{I}^-} : \mathcal{I}^- \to \mathbb{R}$ be an axisymmetric function in $\cap_{s=1}^{\infty}\dot{H}^s\left(\mathbb{R}\times \mathbb{S}^2\right)$ which is supported in $\{v \geq 1\}$.

We may appeal to Theorem~\ref{scatexist} to produce the unique smooth solution $\psi : \mathcal{D}_{\rm ext} \to \mathbb{R}$ to the wave equation such that the radiation field along $\mathcal{I}^-$~\eqref{iplusrad} is given by $\upphi_{\mathcal{I}^-}$ and such that the radiation field along $\mathcal{H}^-$~\eqref{hplusrad} vanishes.

Then $\psi$ extends uniquely as a smooth solution of the wave equation to $\mathcal{D}$ which vanishes along $\mathcal{H}^+_{\rm in}$ (recall that $\mathcal{D}$ is defined by~\eqref{mathcalD}), $\partial_v\psi$ does \underline{not} vanish identically along $\mathcal{CH}^+_{\rm out}$, and we have the following degenerate energy bound along $\mathcal{CH}^+_{\rm out}$:
\begin{equation}\label{degenecauchy}
\int_{-\infty}^{\infty}\int_{\mathbb{S}^2}\left(\partial_v\psi|_{\mathcal{CH}^+_{\rm out}}\right)^2\, dv\, d\mathbb{S}^2 \leq C\int_1^{\infty}\int_{\mathbb{S}^2}\left(\partial_v\upphi_{\mathcal{I}^-}\right)^2\, dv\, d\mathbb{S}^2.
\end{equation}

Finally, the transmission map, $\upphi_{\mathcal{I}^-} \mapsto \psi|_{\mathcal{CH}^+_{\rm out}}$, is time-translation invariant in the sense that if $\psi\left(v,r,\theta,\phi^*\right)$ is the unique  solution which vanishes on $\mathcal{H}^-\cup\mathcal{B}_+\cup\mathcal{H}^+_{\rm in}$ and has the radiation field $\upphi_{\mathcal{I}^-}\left(v,\theta,\phi^*\right)$ along $\mathcal{I}^-$, then, for every $c \in \mathbb{R}$, $\psi\left(v-c,r,\theta,\phi^*\right)$ is the unique  solution which vanishes on $\mathcal{H}^-\cup\mathcal{B}_+\cup\mathcal{H}^+_{\rm in}$ and has the radiation field $\upphi_{\mathcal{I}^-}\left(v-c,\theta,\phi^*\right)$ along $\mathcal{I}^-$.
\end{oldtheorem}
\begin{remark}On the Reissner--Nordstr\"{o}m spacetime, the work~\cite{anne} showed that if one considers solutions $\psi$ whose radiation fields $\uppsi_{\mathcal{H}^+}$ decay polynomially, then a polynomial weight $v^p$ can be added in the integral on the left hand side of~\eqref{degenecauchy}. This strengthened estimate is essential for establishing continuous extendibility of $\psi$ to the Cauchy horizon.
\end{remark}

As we noted in the introduction, it would be of significant interest to extend this preliminary result to a full treatment of interior scattering. In particular, it would be desirable to establish bounded isomorphisms between suitable Hilbert spaces of radiation fields along $\mathcal{H}^+$ and $\mathcal{CH}^+$ (without the assumption of axisymmetry).
\subsubsection{Proof of the estimate~\eqref{degenecauchy}}
The proof of~\eqref{degenecauchy} will in fact follow from a straightforward adaption of estimates from~\cite{anne,lukoh}, so our presentation of that part of the proof will be brief.

The following lemma is an immediate consequence of Theorem~\ref{betterscatexist}, Remark~\ref{finitespeedrmk} and Proposition~\ref{totheinterior}.
\begin{lemma}\label{smooth} Let $\psi$ be as in the statement of Theorem~\ref{tothecauchyandbeyond}. Then $\psi$ extends uniquely to the region $\mathcal{D}$ as a smooth solution of the wave equation which vanishes along $\mathcal{H}^+_{\rm in}$.
\end{lemma}

Next (cf.~\cite{anne}), it immediately follows from the estimate associated to the red-shift vector field~\cite{lectnotes} that we can obtain a non-degenerate energy boundedness statement on constant-$r$ hypersurfaces sufficiently close to $\mathcal{H}^+_{\rm out}$.

\begin{lemma}\label{redshift}Let $\psi : \mathcal{D} \to \mathbb{R}$ be a smooth axisymmetric solution to the wave equation which vanishes for $v$ sufficiently negative, and such that $\psi|_{\mathcal{H}^+} \in \dot{H}^1\left(\mathbb{R} \times \mathbb{S}^2\right)$. Let $c$ satisfy $r_+ -\epsilon \leq c \leq r_+$ for $\epsilon > 0$ sufficiently small. Then
\begin{align*}
\int_{-\infty}^{\infty}\int_{\mathbb{S}^2}&\left[\left(\partial_v\psi\right)^2 + \left(\partial_r\psi\right)^2+\left(\partial_{\theta}\psi\right)^2\right]|_{r = c}\, dv\, d\mathbb{S}^2
\leq C\int_{-\infty}^{\infty}\int_{\mathbb{S}^2}\left[\left(\partial_v\psi\right)^2 + \left(\partial_{\theta}\psi\right)^2\right]|_{\mathcal{H}^+}\, dv\, d\mathbb{S}^2.
\end{align*}
\end{lemma}

For the next two lemmas it will be convenient to switch to $(t,r^*,\theta,\varphi) \in \mathbb{R} \times \mathbb{R} \times \mathbb{S}^2$ coordinates in $\mathcal{D}_{\rm int}$,  where $t$, $\theta$, $\varphi$ are given by their Boyer--Lindquist values (see Section~\ref{blcoord}) and $r^*$ is defined by~\eqref{formrstar}. In these coordinates, the wave operator applied to an axisymmetric function $\psi$ is given by the following formula:
\begin{align}\label{boxform}
\Box_g\psi = \left(\frac{a^2\sin^2\theta\Delta - \left(r^2+a^2\right)^2}{\rho^2\Delta}\right)\partial_t^2\psi
+ \frac{r^2+a^2}{\Delta\rho^2}\partial_{r^*}\left((r^2+a^2)\partial_{r^*}\psi\right) + \frac{1}{\rho^{2}\sin\theta}\partial_{\theta}\left(\sin\theta\partial_{\theta}\psi\right).
\end{align}
Note also that the volume form in $(t,r^*,\theta,\varphi)$ coordinates is given by
\[dVol = \frac{-\rho^2\Delta\sin^2\theta}{r^2+a^2}\, dt\, dr^*\, d\theta\, d\varphi.\]

The next lemma shows that given a non-degenerate energy bound on any constant $\{r^* = c_1\}$ hypersurface, we may also obtain a non-degenerate energy bound on any other constant $\{r^* = c_2\}$ hypersurface via finite-in-time energy estimates (with a constant which blows up as $|c_1|+|c_2| \to\infty$).
\begin{lemma}\label{finitetime}Let $\psi : \mathcal{D}_{\rm int} \to \mathbb{R}$ be a smooth axisymmetric solution to the wave equation. Let $c_2, c_1 \in \mathbb{R}$. Then
\begin{align}\label{someinequality}
\int_{-\infty}^{\infty}\int_{\mathbb{S}^2}&\left[\left(\partial_t\psi\right)^2 + \left(\partial_{r^*}\psi\right)^2 + \left(\partial_{\theta}\psi\right)^2\right]|_{\{r^* = c_2\}}\, dt\, d\mathbb{S}^2
\\ \nonumber &\leq C\left(c_2,c_1\right)\int_{-\infty}^{\infty}\int_{\mathbb{S}^2}\left[\left(\partial_t\psi\right)^2 + \left(\partial_{r^*}\psi\right)^2 + \left(\partial_{\theta}\psi\right)^2\right]|_{\{r^* = c_1\}}\, dt\, d\mathbb{S}^2 ,
\end{align}
where the convention is that~\eqref{someinequality} automatically holds if the right hand side is infinite. Also, we emphasize that the constant $C$ depends on both $c_1$ and $c_2$.
\end{lemma}
\begin{proof}This is an immediate consequence of finite-in-time energy estimates with the timelike time translation invariant vector field $\frac{(r^2+a^2)}{\rho\sqrt{\Delta}}\partial_{r^*}$.
\end{proof}
\begin{remark}Note that the lemma does not require $c_2 < c_1$ or $c_1 < c_2$.
\end{remark}

Next, following closely the approaches from~\cite{anne,lukoh}, for sufficiently large $r^*$, we may use a multiplier of the form $r^q\partial_{r^*}$ to obtain a degenerate energy estimate with a constant which does not blow up as $r^* \to \infty$.
\begin{lemma}\label{rq}Let $\psi : \mathcal{D}_{\rm int} \to \mathbb{R}$ be a smooth axisymmetric solution to the wave equation, and let $c_1 > 0$ be sufficiently large. Then $c_2 > c_1$ implies that
\begin{align}\label{someinequality2}
\int_{-\infty}^{\infty}\int_{\mathbb{S}^2}&\left[\left(\partial_t\psi\right)^2 + \left(\partial_{r^*}\psi\right)^2 -\Delta \left(\partial_{\theta}\psi\right)^2\right]|_{\{r^* = c_2\}}\, dt\, d\mathbb{S}^2
\\ \nonumber &\leq C\int_{-\infty}^{\infty}\int_{\mathbb{S}^2}\left[\left(\partial_t\psi\right)^2 + \left(\partial_{r^*}\psi\right)^2 - \Delta\left(\partial_{\theta}\psi\right)^2\right]|_{\{r^* = c_1\}}\, dt\, d\mathbb{S}^2 ,
\end{align}
where the convention is that~\eqref{someinequality2} automatically holds if the right hand side is infinite. 

(Note that $\Delta < 0$ in $\mathcal{D}_{\rm int}$.) We also emphasize that the constant $C$ does not depend on $c_1$ and $c_2$.
\end{lemma}
\begin{proof}Multiplying the wave equation by $r^q\partial_{r^*}\psi$, integrating over the region in between $\{r^* = c_1\}$ and $\{r^* = c_2\}$ and then integrating by parts eventually yields the boundary terms
\begin{align*}
\left(\int_{r^* = c_2}-\int_{r^*=c_1}\right)&\Bigg[r^q\left(r^2+a^2\right)\left(\partial_{r^*}\psi\right)^2
\\ \nonumber &+ \left(r^q\frac{(r^2+a^2)^2-a^2\sin^2\theta\Delta}{r^2+a^2}\right)\left(\partial_t\psi\right)^2
 - \frac{r^q\Delta}{r^2+a^2}\left(\partial_{\theta}\psi\right)^2\Bigg]\, dt\, d\mathbb{S}^2,
\end{align*}
and a bulk term
\begin{align*}
\int\int\Bigg[&\left[r^q\partial_{r^*}\left(r^2+a^2\right) - \partial_{r^*}\left(r^q\right)\left(r^2+a^2\right)\right]\left(\partial_{r^*}\psi\right)^2\\ \nonumber  &-\partial_{r^*}\left(r^q\frac{(r^2+a^2)^2-a^2\sin^2\theta\Delta}{r^2+a^2}\right)\left(\partial_t\psi\right)^2
+ \partial_{r^*}\left(\frac{r^q\Delta}{r^2+a^2}\right)\left(\partial_{\theta}\psi\right)^2\Bigg]\, dt\, dr^*\, d\mathbb{S}^2.
\end{align*}

If $c_1$ is sufficiently large, then $\partial_{r^*}\Delta > 0$, and, for sufficiently large $q$, it is clear that the bulk is positive. The desired estimate immediately follows.
\end{proof}

Finally, we are ready to prove the estimate~\eqref{degenecauchy}.
\begin{proposition}\label{firsthalf}Let $\psi$ satisfy the hypothesis of Theorem~\ref{tothecauchyandbeyond}, then~\eqref{degenecauchy} holds.
\end{proposition}
\begin{proof}First of all, it follows immediately from Theorem~\ref{betterscatexist} and Lemmas~\ref{smooth},~\ref{redshift},~\ref{finitetime} and~\ref{rq} that $\psi$ extends smoothly to $\mathcal{D}$, vanishes for sufficiently negative $v$, and that for every $c \in \mathbb{R}$ we have
\begin{equation}\label{boundwegot}
\int_{-\infty}^{\infty}\int_{\mathbb{S}^2}\left(\partial_{r^*}\psi\right)^2|_{r^* = c}\, dt\, d\mathbb{S}^2 \leq C\int_1^{\infty}\int_{\mathbb{S}^2}\left(\partial_v\upphi_{\mathcal{I}^-}\right)^2\, dv\, d\mathbb{S}^2,
\end{equation}
for a universal constant $C$.

Next, using that $\psi$ is axisymmetric, we note that a straightforward calculation shows that $\partial_{r^*}$ smoothly extends to $\mathcal{CH}^+_{\rm out}$, and we have $\partial_{r^*}\psi|_{\mathcal{CH}^+_{\rm out}} = \partial_v\psi|_{\mathcal{CH}^+_{\rm out}}$. In turn this is easily seen to imply that for every $v_1 \in (-\infty,\infty)$.
\[\int_{-\infty}^{v_1}\int_{\mathbb{S}^2}\left(\partial_v\psi|_{\mathcal{CH}^+_{\rm out}}\right)^2\, dv\, d\mathbb{S}^2 = \lim_{c\to\infty}\int_{-\infty}^{v_1-c}\int_{\mathbb{S}^2}\left(\partial_{r^*}\psi\right)^2|_{r^*=c}\, dt\, d\mathbb{S}^2.\]
In particular, we see that~\eqref{boundwegot} immediately implies the estimate~\eqref{degenecauchy}.

\end{proof}

\subsubsection{Non-Zero Transmission to $\mathcal{CH}^+_{\rm out}$}\label{transmittocauchy}
We will now show that if $\psi$ satisfies the hypothesis of Theorem~\ref{tothecauchyandbeyond}, then $\partial_v\psi$ cannot vanish identically on $\mathcal{CH}^+_{\rm out}$.

The following straightforward lemma (whose proof we omit) shows that along any constant $r^*$-hypersurface, solutions $\psi$ satisfying the hypothesis of Theorem~\ref{tothecauchyandbeyond} may be well approximated by compactly supported functions.
\begin{lemma}\label{approxcomp}Let $f(v,\theta,\varphi^*) \in C^{\infty}\left(\mathbb{R}\times\mathbb{S}^2\right)$ be axisymmetric, ${\rm supp}(f) \subset \{v \geq v_0\}$ for some $v_0 \in \mathbb{R}$, and \[\int_{v_0}^{\infty}\int_{\mathbb{S}^2}\left[\left(\partial_vf\right)^2 + \left(\partial_{\theta} f\right)^2\right]\, dv\, d\mathbb{S}^2 < \infty.\]
Then there exists $\{f_i\}_{i=1}^{\infty}$ with $f_i \in C^{\infty}_c\left(\mathbb{R}\times\mathbb{S}^2\right)$ such that
\begin{equation}\label{itsdense}
\lim_{i\to\infty}\int_{v_0}^{\infty}\int_{\mathbb{S}^2}\left[\left(\partial_v(f-f_i)\right)^2 + \left(\partial_{\theta} (f-f_i)\right)^2\right]\, dv\, d\mathbb{S}^2 = 0.
\end{equation}
\end{lemma}

Now we are ready for the following proposition.
\begin{proposition}\label{secondhalf}Let $\psi$ satisfy the hypothesis of Theorem~\ref{tothecauchyandbeyond}. Then $\partial_v\psi$ does not vanish identically along $\mathcal{CH}^+_{\rm out}$.
\end{proposition}
\begin{proof}We begin by observing that if $\tilde\psi$ is any solution to the wave equation in $\mathcal{D}_{\rm int}$ with compactly supported Cauchy data along a constant $r^*$-hypersurface, then it follows immediately from the finite speed of propagation and finite-in-time energy estimates that $\tilde\psi$ smoothly extends to $\overline{\mathcal{D}}$ (cf.~Lemma~\ref{smooth}).

Next, we consider the conserved current,
\[\mathbf{J}^T_{\mu}\left[\psi\right] \doteq \mathbf{T}_{\mu\nu}\left[\psi\right]T^{\nu},\]
associated to the Killing vector field $T$ (see Section~\ref{basic}). Keeping in mind that $T+\frac{a}{2Mr_-}Z$ is null and future oriented on $\mathcal{CH}^+_{\rm out}$ and null and \underline{past oriented} on $\mathcal{CH}^+_{\rm in}$, for any axisymmetric solution $\tilde\psi$ to the wave equation in $\mathcal{D}_{\rm in}$ with compactly supported Cauchy data along a constant $r^*$-hypersurface, straightforward calculations yield
\begin{align}\label{signofT}
\int_{\mathcal{CH}^+_{\rm out}}&\mathbf{J}^T_{\mu}\left[\tilde\psi\right]n^{\mu}_{\mathcal{CH}^+_{\rm out}} + \int_{\mathcal{CH}^+_{\rm in}}\mathbf{J}^T_{\mu}\left[\tilde\psi\right]n^{\mu}_{\mathcal{CH}^+_{\rm in}}
\\ \nonumber &= \int_{-\infty}^{\infty}\int_{\mathbb{S}^2}\left(\partial_v\tilde\psi|_{\mathcal{CH}^+_{\rm out}}\right)^2\, dv\, d\mathbb{S}^2 - \int_{-\infty}^{\infty}\int_{\mathbb{S}^2}\left(\partial_u\tilde\psi|_{\mathcal{CH}^+_{\rm in}}\right)^2\, dv\, d\mathbb{S}^2.
\end{align}

In particular, for any $c \in (r_-,r_+)$, the estimate below follows from a straightforward density argument using Lemma~\ref{approxcomp}, the identity~\eqref{signofT}, the divergence theorem and the fact that $\nabla^{\mu}\mathbf{J}^T_{\mu} = 0$:
\begin{equation}\label{Tcauchy}
\int_{-\infty}^{\infty}\int_{\mathbb{S}^2}\left(\partial_v\psi|_{\mathcal{CH}^+_{\rm out}}\right)^2\, dv\, d\mathbb{S}^2 \geq \int_{\{r=c\}}\mathbf{J}^T_{\mu}\left[\psi\right]n^{\mu}_{\{r=c\}}.
\end{equation}
As in~\eqref{signofT}, the integration on the right hand side of this equation is with respect to the induced volume form.

For every $c$ and $v_0$, let
\[\Sigma_{c,v_0} \doteq \{(v,r,\theta,\varphi) : v = v_0{\ \rm and\ }r \in [c,r_+)\}.\]

Now, it follows from an another $\mathbf{J}^T$ energy estimate and a straightforward calculation that for any $v_0 < \infty$
\begin{align}\label{Test2}
\int_{\{r=c\} \cap \{v \leq v_0\}}\mathbf{J}^T_{\mu}\left[\psi\right]n^{\mu}_{\{r=c\}} + \int_{\Sigma_{c,v_0}}\mathbf{J}^T_{\mu}\left[\psi\right]n^{\mu}_{\Sigma_{c,v_0}}
 &= \int_{\mathcal{H}^+\cap \{v \leq v_0\}}\mathbf{J}^T_{\mu}\left[\psi\right]n^{\mu}_{\mathcal{H}^+}
\\ \nonumber &= \int_{-\infty}^{v_0}\int_{\mathbb{S}^2}\left(\partial_v\uppsi_{\mathcal{H}^+}\right)^2\, dv\, d\mathbb{S}^2.
\end{align}

Using the red-shift, one can easily show (cf.~Section 5 of~\cite{anne}) that if $c$ is taken sufficiently close to $r_+$, then
\[\liminf_{v_0\to\infty}\int_{\Sigma_{c,v_0}}\mathbf{J}^T_{\mu}\left[\psi\right]n^{\mu}_{\Sigma_{c,v_0}} = 0.\]
Thus,~\eqref{Tcauchy} and~\eqref{Test2} together imply
\begin{equation}\label{Tshowsitisthere}
\int_{-\infty}^{\infty}\int_{\mathbb{S}^2}\left(\partial_v\psi|_{\mathcal{CH}^+_{\rm out}}\right)^2\, dv\, d\mathbb{S}^2 \geq \int_{-\infty}^{\infty}\int_{\mathbb{S}^2}\left(\partial_v\uppsi_{\mathcal{H}^+}\right)^2\, dv\, d\mathbb{S}^2.
\end{equation}

Since Theorem~\ref{scatexist} implies the right hand side of~\eqref{Tshowsitisthere} is strictly positive, the proposition is proved.
\end{proof}

Combining Propositions~\ref{firsthalf}, Proposition~\ref{secondhalf}, Proposition~\ref{totheinterior} and Proposition~\ref{timetranslate} yields Theorem~\ref{tothecauchyandbeyond}.

\subsection{Proof of Theorem~\ref{IminusToCauchy}}\label{proofofthe2}
We are now ready to prove Theorem~\ref{IminusToCauchy}. The proof will be close in spirit to the proof of Theorem~\ref{IplusToEvent}.
First,  in Section~\ref{constructsomuch} we construct the radiation field $\upphi_{\mathcal{I}^-}$  and then show 
in Section~\ref{prooftheo2} that the following integral of $\psi$ blows up along $\mathcal{CH}^+_{\rm out}$:
\begin{equation}\label{thisblows}
\int_{-\infty}^{\infty}\int_{\mathbb{S}^2}\left(\partial_v\left(\psi\right)|_{\mathcal{CH}^+_{\rm out}}\right)^2e^{-\kappa_-v}\, dv\, d\mathbb{S}^2  = \infty.
\end{equation}
Keeping~\eqref{formulabifen2} in mind, one immediately sees that~\eqref{thisblows} implies that the local energy of $\psi$ is infinite along $\mathcal{H}^-$ near the bifurcation sphere. (Note, however, that Theorem~\ref{tothecauchyandbeyond} implies that $\psi$ has a finite $\mathbf{J}^T$-energy along $\mathcal{CH}^+_{\rm out}$.)

Finally, in Section~\ref{propsing} we show, by an elementary propagation of singularities argument, 
that this blow-up of local energy implies Theorem~\ref{IminusToCauchy}.

\subsubsection{Construction of $\upphi_{\mathcal{I}^-}$}\label{constructsomuch}
We turn to the construction of the function $\upphi_{\mathcal{I}^-}$, which will be very closely related to our construction in Section~\ref{constructrad}. Let $\left(\upphi_0\right)_{\mathcal{I}^-}(v,\theta,\varphi^*) : \mathcal{I}^- \to \mathbb{R}$ be a non-zero smooth axisymmetric function which is compactly supported in $\{v \in (1,2)\}$. Then, using Theorem~\ref{tothecauchyandbeyond}, let $\psi_0 : \mathcal{D} \to \mathbb{R}$ be the unique solution to the wave equation with a vanishing radiation field along $\mathcal{H}^-$ and with radiation field $\left(\upphi_0\right)_{\mathcal{I}^+}$ along $\mathcal{I}^+$.

Recall that Theorem~\ref{tothecauchyandbeyond} guarantees that $\partial_v\left(\psi_0|_{\mathcal{CH}^+_{\rm out}}\right)$ is a smooth function and does not vanish identically. In particular, after rescaling $\left(\upphi_0\right)_{\mathcal{I}^-}$, we can assume without loss of generality that
\begin{equation}\label{normalize2}
\int_{-\infty}^{\infty}\int_{\mathbb{S}^2}\left(\partial_v\left(\psi_0|_{\mathcal{CH}^+_{\rm out}}\right)\right)^2\, dv\, d\mathbb{S}^2 = 1.
\end{equation}

We now consider two separate cases: If~\eqref{thisblows} holds with $\psi$ replaced by $\psi_0$, then we simply set $\psi = \psi_0$ and $\left(\upphi_{\mathcal{I}^-}\right)_{\mathcal{I}^-} = \left(\upphi_0\right)_{\mathcal{I}^-}$. 

If~\eqref{thisblows} does not hold with $\psi$ replaced by $\psi_0$, then we may assume that
\begin{equation}\label{quitetheassumption}
E \doteq \int_{-\infty}^{\infty}\int_{\mathbb{S}^2}\left(\partial_v\left(\psi_0\right)|_{\mathcal{CH}^+_{\rm out}}\right)^2e^{-\kappa_-v}\, dv\, d\mathbb{S}^2  < \infty.
\end{equation}

Now, recalling the number $p$ from Theorem~\ref{IminusToCauchy}, define
\[\left(\upphi_i\right)_{\mathcal{I}^-}\left(v,\theta,\varphi^*\right) \doteq i^{-2p}\left(\upphi_0\right)_{\mathcal{I}^-}\left(v - i^2,\theta,\varphi\right),\qquad \upphi_{\mathcal{I}^-} \doteq \sum_{i=0}^{\infty}\left(\upphi_i\right)_{\mathcal{I}^-}.\]
Note in particular that $\upphi_{\mathcal{I}^-}$ satisfies the requirements from Theorem~\ref{IminusToCauchy}, i.e.,~it is axisymmetric and we may easily check that
\[\left|T^i\partial_{\theta}^j\upphi_{\mathcal{I}^-}\right| \leq C_{i,j}\left(1+\left|v\right|\right)^{-p}\qquad \forall i,j \geq 0.\]

Finally, using Theorem~\ref{tothecauchyandbeyond}, we let $\psi_i:\mathcal{D} \to \mathbb{R}$ and $\psi : \mathcal{D} \to \mathbb{R}$ denote the unique solutions to the wave equation with vanishing radiation fields along $\mathcal{H}^-\cup\mathcal{B}_+\cup\mathcal{H}^+_{\rm in}$ and radiation fields along $\mathcal{I}^-$ given by $\left(\upphi_i\right)_{\mathcal{I}^-}$ and $\upphi_{\mathcal{I}^-}$ respectively.

Just as in Section~\ref{constructrad}, we note the following fundamental relation:
\begin{lemma}\label{thewholepoint2}
\begin{equation*}
\psi_i\left(v,\theta,\varphi^*\right) = i^{-2p}\psi_0\left(v-i^2,\theta,\varphi\right).
\end{equation*}
\end{lemma}
\begin{proof}This is proved in the same fashion as Lemma~\ref{thewholepoint} in view of the time-translation property stated in Theorem~\ref{tothecauchyandbeyond}.
\end{proof}

\subsubsection{Almost orthogonality and the proof of~\eqref{thisblows} }\label{prooftheo2}
Now we establish that $\psi$ satisfies the blow-up~\eqref{thisblows}. First of all, it is clear that without loss of generality, we may assume that~\eqref{quitetheassumption} holds.

In what follows we will omit the ``$|_{\mathcal{CH}^+_{\rm out}}$'', which should appear each time we write $\psi$ and $\psi_i$, and the spherical volume forms $d\mathbb{S}^2$.

Let $i \in \mathbb{Z}_{> 0}$. We have
\begin{align}\label{thisishowitbegins}
\int_{1+i^2}^{\infty}\int_{\mathbb{S}^2}\left(\partial_v\psi\right)^2\, dv &\geq \frac{1}{2}\int_{1+i^2}^{\infty}\int_{\mathbb{S}^2}\left(\sum_{j=i}^{\infty}\partial_v\psi_j\right)^2\, dv - 4\int_{1+i^2}^{\infty}\int_{\mathbb{S}^2}\left(\sum_{j=0}^{i-1}\partial_v\psi_j\right)^2\, dv
\\ \nonumber &\doteq I - II.
\end{align}

It easily follows from Lemma~\ref{thewholepoint2} and~\eqref{quitetheassumption} that there exists a constant $C$, depending on $\psi_0$ but independent of $i$, such that
\begin{equation}\label{wecanignoreII}
II \leq Cie^{\kappa_-i}.
\end{equation}

Next, we note that
\begin{align}\label{Iprime}
I \geq \frac{1}{2}\sum_{j=i}^{\infty}\int_{1+i^2}^{\infty}\int_{\mathbb{S}^2}\left(\partial_v\psi_j\right)^2\, dv - \sum_{j=i}^{\infty}\sum_{k=j+1}^{\infty}\int_{1+i^2}^{\infty}\int_{\mathbb{S}^2}\left|\partial_v\psi_j\partial_v\psi_k\right|\, dv \doteq I' - I''.
\end{align}

Using Lemma~\ref{thewholepoint2} and~\eqref{quitetheassumption}, it is straightforward to establish that $i \leq j < k$ implies
\begin{align}\label{orthoest}
\int_{1+i^2}^{\infty}\int_{\mathbb{S}^2}1_{{\rm supp}\left(\partial_v\psi_k\right)}\left(\partial_v\psi_j\right)^2\, dv &\leq Ee^{\kappa_-\left(k^2-j^2\right)}\int_{1+i^2}^{\infty}\int_{\mathbb{S}^2}\left(\partial_v\psi_j\right)^2\, dv
\\ \nonumber &\leq Ee^{\kappa_-\left(k+j\right)}\int_{1+i^2}^{\infty}\int_{\mathbb{S}^2}\left(\partial_v\psi_j\right)^2\, dv.
\end{align}
In particular, for $i$ sufficiently large,~\eqref{orthoest} yields
\begin{align}\label{estimateIdoubleprime}
I'' &\leq \sqrt{E}\sum_{j=i}^{\infty}\sum_{k=j+1}^{\infty}e^{\frac{\kappa_-}{2}(j+k)}\sqrt{\int_{1+i^2}^{\infty}\int_{\mathbb{S}^2}\left(\partial_v\psi_j\right)^2\, dv}\sqrt{\int_{1+i^2}^{\infty}\int_{\mathbb{S}^2}\left(\partial_v\psi_k\right)^2\, dv}
\\ \nonumber &\leq \frac{1}{4}\sum_{j=1}^{\infty}\int_{1+i^2}^{\infty}\int_{\mathbb{S}^2}\left(\partial_v\psi_j\right)^2\, dv.
\end{align}

For $i$ sufficiently large, combining~\eqref{Iprime} with~\eqref{estimateIdoubleprime} and then using Lemma~\ref{thewholepoint2} yields
\[I \geq \frac{1}{4}\sum_{j=i}^{\infty}\int_{1+i^2}^{\infty}\int_{\mathbb{S}^2}\left(\partial_v\psi_j\right)^2\, dv \geq \frac{i^{-2p}}{4}.\]
Combining this with~\eqref{thisishowitbegins} and~\eqref{wecanignoreII} and taking $i$ sufficiently large finally yields
\[\int_{1+i^2}^{\infty}\int_{\mathbb{S}^2}\left(\partial_v\psi\right)^2\, dv \geq \frac{i^{-2p}}{8}.\]
This immediately implies that
\[\int_1^{\infty}\int_{\mathbb{S}^2}\left(\partial_v\psi\right)^2v^{2p}\, dv = \infty,\]
and hence we obtain~\eqref{thisblows}.

\subsubsection{Propagation of singularities}\label{propsing}
The following lemma will be proven by a general propagation of singularities argument. (This can be contrasted with the less standard argument used in Section~\ref{notstandardprop}.)
\begin{lemma}\label{thissufficies}Let $\psi : \mathcal{D} \to \mathbb{R}$ be a smooth solution to the wave equation supported in $\{v \geq v_0\}$, for some $v_0 \in \mathbb{R}$, such that
\begin{equation}\label{thisisinfinite}
\int_{-\infty}^{\infty}\int_{\mathbb{S}^2}\left(\partial_v\psi|_{\mathcal{CH}^+_{\rm out}}\right)^2e^{-\kappa_-v}\, dv\, d\mathbb{S}^2 = \infty,
\end{equation}
then $\psi$ does not lie in $H^1_{\rm loc}$ of any neighbourhood of any point on $\overline{\mathcal{CH}^+_{\rm in}}$.
\end{lemma}
\begin{proof}
Keeping~\eqref{formulabifen2} in mind, we note the local energy along $\mathcal{CH}^+_{\rm out}$ is proportional to
\[\int_{v_0}^{\infty}\int_{\mathbb{S}^2}\left[\left(\partial_v\psi\right)^2 + \left(\partial_{\theta}\psi\right)^2\right]|_{\mathcal{CH}^+_{\rm out}}e^{-\kappa_-v}\, dv\, d\mathbb{S}^2;\]
in particular, the local energy is infinite.

In view of the support of $\psi$, finite-in-time energy estimates in the region $J^+\left(\tilde{\Sigma} \cap \{v \geq v_0\}\right)$ immediately imply that $\psi$ does not lie in $H^1_{\rm loc}\left(\tilde\Sigma\right)$ for any suitably regular spacelike hypersurface $\tilde\Sigma$ intersecting $\mathcal{CH}^+_{\rm in}$ transversally. 
\begin{center}
\begin{tikzpicture}
\fill[lightgray] (-2,2)--(0,0)--(-2,-2) -- (-4,0); 
\fill[gray] (-2,2) -- (-1.2,1.2) to [out = 195, in = 15] (-3.25,.25) -- (-3.5,.5) --(-2,2);

\draw (-2,-2) -- (-4,0) node[sloped,below,midway]{$\mathcal{H}^+_{\rm in}$}; 
\draw (-2,-2) -- (0,0) node[sloped,below,midway]{$\mathcal{H}^+_{\rm out}$}; 
\draw (-4,0) -- (-2,2) node[sloped,above,midway]{$\mathcal{CH}^+_{\rm out}$}; 
\draw (-2,2) -- (0,0) node[sloped,above,midway]{$\mathcal{CH}^+_{\rm in}$}; 
\path [draw=black,fill=white] (0,0) circle (1/16); 

\path [draw=black,fill=black] (-2,-2) circle (1/16) node[left]{$\mathcal{B}_+$}; 
\path [draw=black,fill=white] (-4,0) circle (1/16); 
\path [draw=black,fill=black] (-2,2) circle (1/16) node[left]{$\mathcal{B}_-$}; 

\draw (-2,2) -- (-1.2,1.2) to [out = 195, in = 15] node[sloped,below,midway]{$\tilde\Sigma$} (-3.25,.25) -- (-3.5,.5) --(-2,2);

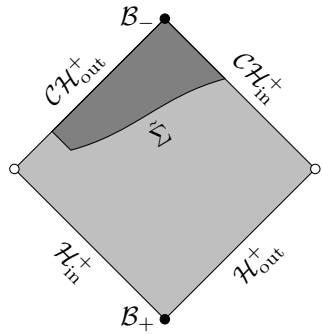
\node [align = flush center] at (-2,-2.5) {Figure 12: The region of finite-in-time energy estimates};
\end{tikzpicture}
\end{center}

It immediately follows that $\psi$ does not lie in $H^1_{\rm loc}$ of any neighbourhood of any point on $\overline{\mathcal{CH}^+_{\rm in}}$.
\end{proof}

Combining Lemma~\ref{thissufficies} with the fact that Sections~\ref{constructsomuch} and~\ref{prooftheo2} have constructed radiation fields $\upphi_{\mathcal{I}^-}$ leading to $\psi$ satisfying~\eqref{thisblows}, concludes the proof of Theorem~\ref{IminusToCauchy}.

\end{document}